\documentclass[11pt]{article}
\usepackage{graphicx,amsmath,amssymb,mathrsfs}
\usepackage[a4paper,  margin=1in]{geometry}
\usepackage[utf8]{inputenc}
\usepackage[english]{babel}
\usepackage[utf8]{inputenc}
\usepackage{amsmath}
\usepackage{longtable}
\usepackage{subcaption}
\usepackage{amsfonts}
\usepackage{amssymb} 
\usepackage{booktabs,multirow}
\usepackage{float}
\usepackage{algorithm}
\usepackage{algorithmic}
\usepackage{enumitem} 
\usepackage{authblk}
\usepackage{natbib}
\usepackage[none]{hyphenat}
\usepackage{setspace}
\usepackage{tcolorbox}
\setcitestyle{authoryear,open={(},close={)}}

\newtheorem{theorem}{Theorem}[section]
\newtheorem{corollary}{Corollary}[theorem]
\newtheorem{lemma}[theorem]{Lemma} 

\usepackage{hyperref}
\hypersetup{
    colorlinks=true,
    linkcolor=blue,
    filecolor=magenta,      
    urlcolor=cyan,
    citecolor=blue
    }
\usepackage{verbatim}
\immediate\write18{texcount -tex -sum  \jobname.tex > \jobname.wordcount.tex}

\title{\textbf{Estimation of the complexity of a network under a Gaussian graphical model}}

\author{
Nabaneet Das$^{1}$\thanks{Email: \texttt{nabaneet@uni-bremen.de}} , 
Thorsten Dickhaus$^{2}$\thanks{Email: \texttt{dickhaus@uni-bremen.de}} \\
\small $^{1,2}$ Institute for Statistics, University of Bremen, 28359, Bremen, Germany
}

\date{}

\begin{document}

\maketitle 

\begin{abstract}
The proportion of edges in a Gaussian graphical model (GGM) characterizes the complexity of its conditional dependence structure. Since edge presence corresponds to a nonzero entry of the precision matrix, estimation of this proportion can be formulated as a large-scale multiple testing problem. We propose an estimator that combines p-values from simultaneous edge-wise tests, conducted under false discovery rate control, with Storey’s estimator of the proportion of true null hypotheses. We establish weak dependence conditions on the precision matrix under which the empirical cumulative distribution function of the p-values converges to its population counterpart. These conditions cover high-dimensional regimes, including those arising in genetic association studies. Under such dependence, we characterize the asymptotic bias of the Schweder--Spjøtvoll estimator, showing that it is upward biased and thus underestimates the true edge proportion. Simulation studies across a variety of models illustrate the
finite-sample performance of the proposed estimator.
\end{abstract}


\section{Introduction}
Understanding the relationships among multiple variables is an important problem in many areas of science, including biology, finance, and the social sciences. Accurately capturing these relationships helps us infer networks, identify key variables, and understand underlying mechanisms in complex systems. Gaussian Graphical Models (GGMs) are a widely used framework for representing conditional dependencies among jointly Gaussian variables.  \\[0.05 in]
Formally, let $\mathbf{X} = (X_1, \dots, X_k)^T \sim N (\boldsymbol{\mu}, \boldsymbol{\Sigma})$ be a $k$-dimensional multivariate normal vector with mean $\boldsymbol{\mu}$ and covariance matrix $\boldsymbol{\Sigma} $. A GGM represents $\mathbf{X}$ as an undirected graph $G = (V, E)$, where each vertex $i \in V$ corresponds to a variable $X_i$, and an edge $(i,j) \in E$ indicates that $X_i$ and $X_j$ are conditionally dependent given all other variables. Importantly, a well-known result in this framework states that $(i,j) \in E$ if and only if the corresponding entry of the precision matrix $\mathbf{\Omega} = \boldsymbol{\Sigma} ^{-1}$ is non-zero, i.e., $\omega_{ij} \neq 0$ (see \cite{lauritzen1996graphical}). In this way, GGMs link the problem of network inference directly to the estimation of the precision matrix. \\[0.05 in]
In the sequel, we refer to the set \(E\) as the network structure of \(\mathbf{X}\). In practice, this structure is typically unknown or only partially observed, which motivates the use of statistical methods to estimate \(E\). We assume that an i.i.d.\ sample \(\mathbf{X}_1, \dots, \mathbf{X}_n\) is available, with each \(\mathbf{X}_i\) following the same distribution as \(\mathbf{X}\). Estimation of the network structure is particularly challenging in high-dimensional settings, where the number of variables \(k\) may be large relative to the sample size. Traditional approaches often rely on regularized optimization techniques to estimate sparse precision matrices. Popular methods include the graphical Lasso (see \cite{friedman2008sparse}, \cite{d2008first}), Scaled Lasso (\cite{sun2012scaled}), and Dantzig selector-based approaches (\cite{cai2011constrained}), which encourage sparsity and improve estimation in high dimensions. Other methods based on $l_1$-minimization techniques can be found in \cite{meinshausen2006high}, \cite{yuan2010high}, \cite{zhang2010estimation}, \cite{cai2011constrained}, \cite{liu2012high}, and \cite{xue2012regularized}. \\[0.05 in]
An alternative approach, proposed by \cite{liu2013gaussian}, frames Gaussian graphical model (GGM) estimation as a multiple testing problem. In this framework, each pair of variables corresponds to a hypothesis test:
\begin{equation}
H_{0,ij}: \omega_{ij} = 0 \quad \text{versus} \quad H_{1,ij}: \omega_{ij} \neq 0, \quad 1 \le i < j \le k.
\label{MHT}
\end{equation}
There are $k(k-1)/2$ hypotheses to be tested. Owing to the intrinsic structure of the precision matrix, these hypotheses are mutually dependent, and consequently the associated test statistics exhibit complex dependence. Multiple testing under dependence is a challenging problem, and relatively limited literature provides a comprehensive theoretical understanding of the performance of multiple testing procedures in such settings. \\[0.05 in]
Several type I error criteria have been proposed for simultaneous inference. Procedures controlling the familywise error rate (FWER) are well known to be conservative in high-dimensional regimes; see, for example, \cite{das2021bound}, \cite{dey2024limiting}, and the references therein. As a less stringent alternative, \cite{benjamini1995controlling} introduced the false discovery rate (FDR) and proposed a step-up procedure for its control. The behavior of this procedure under various forms of dependence was subsequently investigated in \cite{benjamini2001control}, \cite{sarkar2002some}, and \cite{FDR2007}. The methodology of \cite{liu2013gaussian} explicitly targets FDR control in high-dimensional GGMs. Their approach is inspired by \cite{efron2001empirical} and employs an empirical approximation to the false discovery proportion (FDP) in order to determine an adaptive rejection threshold. The GFC (GGM estimation with FDR control) procedure uses regularized estimators, such as the Lasso, Scaled Lasso, or Dantzig-type estimators, to compute test statistics for each pair of variables. A key feature of the GFC approach is that it establishes the asymptotic normality of these test statistics, which allows for the computation of two-sided p-values for each edge (see \cite{liu2013gaussian}). By treating GGM estimation as a large-scale multiple testing problem, this procedure provides a systematic and statistically principled way to infer the graph structure. Recent developments on FDR-controlled graphical model learning include knockoff-based and reproducibility-oriented procedures; see, for example, \cite{li2021ggm}, \cite{zhou2022reproducible}, and \cite{zhou2026reproducible}. These works further highlight the continuing interest in controlling false discoveries in high-dimensional graphical model selection. \\[0.05 in]
In this work, we consider the proportion of false null hypotheses, which corresponds to the proportion of edges in the graph and thus quantifies the overall complexity of the Gaussian graphical model. This is a global feature of the graph, in contrast to much of the existing literature, which primarily focuses on local features such as node-wise neighborhoods or pairwise conditional dependencies, typically estimated via neighborhood selection or sparse precision matrix methods ( e.g., \cite{meinshausen2006high}, \cite{raskutti2008model}, \cite{friedman2008sparse}). While these approaches are well suited for recovering local graph structure, they do not directly address the estimation of global network characteristics such as edge density or overall graph complexity. \\[0.05 in]
Estimation of the proportion of false null hypotheses is a well-explored topic within the domain of simultaneous statistical inference and serves as a foundation for data-adaptive procedures. One of the most prominent estimation techniques is the graphical approach introduced by \cite{schweder1982plots}, as formalized in \eqref{storey}. Building upon this foundational work, \cite{benjamini2000adaptive} proposed an adaptive version of the Benjamini--Hochberg (BH) procedure. Subsequently, \cite{storey2002direct} and \cite{storey2004strong} suggested various strategies for selecting the tuning parameters of the Schweder--Spj\o{}tvall estimator. For a comprehensive discussion on adaptive procedures based on this estimator and the corresponding control of the false discovery rate (FDR), we refer the reader to Section 3.1.3 and Theorem 5.4 of \cite{dickhaus2014simultaneous}. Further developments focusing on the density estimation of p-values were explored by \cite{langaas2005estimating} and \cite{genovese2004stochastic}. Specifically, \cite{genovese2004stochastic} discussed the consistency of the density-based estimators originally proposed by \cite{swanepoel1999} and \cite{hengartner1999}. More recently, \cite{patra2016estimation} introduced a consistent estimator based on the empirical cumulative distribution function (ECDF) of p-values under a two-component mixture model. In parallel with spatial-domain approaches, a separate line of work considers estimators constructed in the frequency domain. Estimators based on the empirical characteristic function have been developed for normal mixture models in large-scale multiple testing by \cite{jin2007estimating,jin2008proportion}, where both the null and non-null components are assumed to follow normal distributions, and the goal is to estimate null parameters and the proportion of non-null signals. This framework, as developed in \cite{chen2019uniformly}, relaxes the normality assumption for the null while assuming that the alternative belongs to a location-shift family, and establishes uniform consistency of the estimators using Lebesgue–Stieltjes integral equations and harmonic analysis techniques. More recent work, as presented in \cite{chen2025uniformly}, generalizes the approach to composite null hypotheses and broader families of distributions, including the Gamma family, while maintaining the location-shift assumption for the alternative. Frequency-domain methods provide an alternative to p-value–based approaches and allow consistent estimation when the alternative distribution is a shifted version of the null. \\[0.05 in]
While many existing methods for estimating the proportion of false null hypotheses assume independence among p-values, this assumption is often violated in the Gaussian Graphical Model setting due to the inherent dependence among precision matrix entries. Despite this, the Schweder--Spj\o{}tvall estimator remains an attractive choice due to its simplicity, widespread adoption, and its ability to handle weakly dependent p-values. Moreover, it typically provides a conservative estimate, which facilitates the control of the FDR. Its validity rests on the property that the ECDF of the p-values consistently estimates the average CDF under these conditions.
 \\[0.05 in]
In this article, we combine the GFC procedure of \cite{liu2013gaussian} with
the Schweder--Spj{\o}tvoll estimator to estimate the proportion of false null
hypotheses and, consequently, the graph complexity. The resulting estimate is
conservative in the sense that the proportion of true null hypotheses tends to
be overestimated, leading to an underestimation of the proportion of edges in
the graph. P-values for all pairwise conditional dependence tests are first
computed using the GFC procedure. The Schweder--Spj{\o}tvoll estimator is then
applied, together with the tuning parameter selection methodology of
\cite{storey2002direct,storey2003statistical}, to estimate the proportion of
edges in the graph. This approach leverages the advantages of regularized
precision matrix estimation, asymptotic normality of test statistics, and
robust multiple testing procedures. \\[0.05 in]
The remainder of the paper is organized as follows. In Section 2, we describe the problem setup and present the test statistics derived from the GFC procedure of \cite{liu2013gaussian}, along with a brief overview of the Schweder–Spjøtvall estimator. Section 3 establishes the asymptotic validity of the Schweder–Spjøtvall estimator, showing that the ECDF of the p-values consistently estimates the average CDF as long as the sum of the absolute values of the precision matrix entries remains $o(k^2)$. Section 4 presents simulation studies to evaluate the numerical performance of the estimator. Section 5 illustrates the methodology using a real data application, and the appendix presents the proofs of the theoretical results.

\section{Description of the problem}\label{Desc}
We have $n$ observations $\mathbf{X}_1, ..., \mathbf{X}_n$ from a $k$-variate normal distribution with mean vector $\boldsymbol{\mu}$ and covariance matrix $\boldsymbol{\Sigma} $. Without loss of generality, we will assume that $\boldsymbol{\mu} = \mathbf{0}_k$ (a $ \: k \times 1 $ vector consisting of all zeros). We are interested in the multiple testing problem of \eqref{MHT}. We begin by introducing some basic notation. For any vector \(\mathbf{x} \in \mathbb{R}^k\), let \(\mathbf{x}_{-i}\) denote the \((k-1)\)-dimensional vector obtained by removing the \(i\)-th component (\(x_i\)) from \(\mathbf{x} = (x_1, \dots, x_k)'\). We denote the average of the entries of the vector $\mathbf{x}$ by $\bar{\mathbf{x}} $. For any \(p \times q\) matrix \(\mathbf{A}\), let \(\mathbf{A}_{i, -j}\) denote the \(i\)-th row of \(\mathbf{A}\) with its \(j\)-th entry removed, and let \(\mathbf{A}_{-i, j}\) denote the \(j\)-th column of \(\mathbf{A}\) with its \(i\)-th entry removed. Finally, let \(\mathbf{A}_{-i, -j}\) denote the \((p - 1) \times (q - 1)\) matrix obtained by deleting the \(i\)-th row and the \(j\)-th column of \(\mathbf{A}\). For notational simplicity, we suppress the dependence of the test statistics
$T_{ij}$, the corresponding $p$-values $p_{ij}$, and other related quantities
on $n$ and $k$ throughout the manuscript. It is understood that all these
quantities are computed from $n$ independent $k$-variate observations,
$\mathbf{X}_1, \ldots, \mathbf{X}_n$. We consider the testing framework of \cite{liu2013gaussian} where the test statistic $T_{ij} $ for the hypothesis $H_{0,ij}$ is constructed as follows :  
\begin{enumerate}[label=(\Roman*)]
    \item  For \(\mathbf{X} = ( X_{1}, \dots , X_{k} )' \sim N ( \boldsymbol{ \mu} , \boldsymbol{\Sigma} )\), we can write
\[
X_{ i} = \alpha_i + \mathbf{X}'_{ -i } \boldsymbol{\beta}_i + \varepsilon_{ i}, \quad 1 \leq i \leq k,
\]
where \(\varepsilon_i \sim N\left(0,\;
\sigma_{ii} - \boldsymbol{\Sigma}_{i, -i}
\boldsymbol{\Sigma}^{-1}_{-i, -i}
\boldsymbol{\Sigma}_{-i, i}\right)\) is independent of
\(\mathbf{X}_{-i}\),
\(\alpha_i = \mu_i -
\boldsymbol{\Sigma}_{i, -i}
\boldsymbol{\Sigma}^{-1}_{-i, -i}
\boldsymbol{\mu}_{-i}\), and
\((\sigma_{ij})_{k \times k} = \boldsymbol{\Sigma}\). The regression coefficient vector \(\boldsymbol{\beta}_i\) and the error
terms \(\varepsilon_i\) satisfy
\[
\boldsymbol{\beta}_i = -\omega_{ii}^{-1} \boldsymbol{\Omega}_{-i, i}
\quad \text{and} \quad
\mathrm{Cov}(\varepsilon_i, \varepsilon_j)
= \frac{\omega_{ij}}{\omega_{ii}\omega_{jj}}.
\]
We estimate the GGM by recovering the set of non-null entries of \(\Sigma_{\varepsilon}\),
the covariance matrix of \(\boldsymbol{\varepsilon}
= (\varepsilon_1, \ldots, \varepsilon_k)'\). 
\item For \(\mathbf{X}_l = (X_{l1}, \ldots, X_{lk})'\), we can write
\[
X_{li} = \alpha_i +  \boldsymbol{X}_{l,-i} \boldsymbol{\beta}_i + \varepsilon_{li},
\qquad 1 \leq i \leq k , \: \: 1 \leq l \leq n.
\]
The estimators \(\hat{\boldsymbol{\beta}}_i\) are obtained using the Lasso or the scaled
Lasso, following the GFC procedure of \cite{liu2013gaussian}. Throughout the manuscript, these estimators are assumed to
satisfy the regularity conditions on the initial regression estimators
imposed in Proposition~3.1 of \cite{liu2013gaussian}; these conditions
will not be stated explicitly in the subsequent theoretical results.

\item Define the residuals by
\[
\hat{\varepsilon}_{li}
= X_{li} - \bar{\mathbf{X}}_i - (\mathbf{X}_{l,-i} - \bar{\mathbf{X}}_{-i}) \hat{\boldsymbol{\beta}}_i,
\]
and define the sample covariance coefficients between the residuals by
\[
\hat{r}_{ij}
= \frac{1}{n} \sum_{l=1}^{n}
\hat{\varepsilon}_{li} \hat{\varepsilon}_{lj}.
\]

\item Define
\[
T_{1,ij}
= \frac{1}{n} \Bigg(
\sum_{l=1}^{n} \hat{\varepsilon}_{li} \hat{\varepsilon}_{lj}
+ \sum_{l=1}^{n} \hat{\varepsilon}_{li}^2 \hat{\boldsymbol{\beta}}_{i,j}
+ \sum_{l=1}^{n} \hat{\varepsilon}_{lj}^2 \hat{\boldsymbol{\beta}}_{j-1,i}
\Bigg).
\]
It was shown in \cite{liu2013gaussian} that, under certain regularity conditions and assuming $\log k = o(n)$, 
\[
\sqrt{\frac{n}{\hat{r}_{ii} \hat{r}_{jj}}} 
\Bigg(
T_{1,ij} + b_{n,ij} \frac{\omega_{ij}}{\omega_{ii}\omega_{jj}}
\Bigg)
\xrightarrow{d}
N\!\Bigg(
0, 1 + \frac{\omega_{ij}^2}{\omega_{ii}\omega_{jj}}
\Bigg)
\quad \text{as }  \min\{n,k\}\to\infty,
\]
where $b_{n,ij} = \omega_{ii} \hat{\sigma}_{ii,\varepsilon} + \omega_{jj} \hat{\sigma}_{jj,\varepsilon} - 1$, and $\hat{\sigma}_{ij,\varepsilon}$ denotes the $(i,j)$-th element of the matrix 
$\hat{\Sigma}_{\varepsilon} = \frac{1}{n} \sum_{l=1}^{n} (\boldsymbol{\varepsilon}_l - \bar{\boldsymbol{\varepsilon}}) (\boldsymbol{\varepsilon}_l - \bar{\boldsymbol{\varepsilon}})'$. Here,
\(\boldsymbol{\varepsilon}_l
= (\varepsilon_{l1}, \ldots, \varepsilon_{lk})'\) and
\(\bar{\boldsymbol{\varepsilon}}
= \frac{1}{n} \sum_{l=1}^{n} \boldsymbol{\varepsilon}_l\). 
\medskip
\item Finally, the test statistic for testing \(H_{0,ij}\) is defined as
\begin{equation}
T_{ij}
= \sqrt{\frac{n}{\hat{r}_{ii} \hat{r}_{jj}}}\, T_{1,ij},
\qquad 1 \leq i < j \leq k.
\label{TS}
\end{equation}
\end{enumerate}
Under Condition~\eqref{C1}, defined below, it was shown in \cite{liu2013gaussian}
that
\[
T_{ij} \overset{d}{\to} N(0,1)
\: \:  \text{under } H_{0,ij} \quad
\text{ as } \:  \min\{n,k\}\to\infty.
\]
The convergence in distribution is uniform for
\(1 \leq i < j \leq k\). Condition~\eqref{C1} is defined as follows:
\medskip
\begingroup
\renewcommand{\theequation}{C1}
\renewcommand{\theHequation}{C1}
\begin{equation}\label{C1}
\max_{1 \leq i \leq k} \sigma_{ii} \leq c_0,
\qquad
\max_{1 \leq i \leq k} \omega_{ii} \leq c_0,
\qquad
\text{for some } c_0 > 0,
\quad
\log k = o(n).
\end{equation}
\endgroup
\addtocounter{equation}{-1}
\medskip
The GFC procedure of \cite{liu2013gaussian} for testing $H_{0,ij}$ is implemented as follows :  \\[0.05 in]
For $ 0 < \alpha < 1 $ and $T_{ij}$'s defined in \eqref{TS}, define 
\begin{equation}
\hat t_{\alpha}
= \inf \left\{
0 \le t \le 2\sqrt{\log k} :
\frac{ G(t)\, k(k-1)/2 }
{\max\!\left\{1,\; \sum_{1 \le i < j \le k} \mathbf{I}\!\left(|T_{ij}| > t\right)\right\}}
\le \alpha
\right\}.
\label{t_hat}
\end{equation}
where $ G(t) = P( |Z| > t ) $ and $  Z \sim N(0,1) $. If the infimum in \eqref{t_hat} does not exist, then define $\hat{t}_{ \alpha} = 2 \sqrt{\log k}$. \\[0.05 in]
Finally we reject $H_{0,ij} $ if $ | T_{ij} | > \hat{t}_{ \alpha} $. Under the assumptions of Theorem 3.1 of \cite{liu2013gaussian}, the aforementioned procedure satisfies 
$$   \frac{ FDP}{ \alpha \pi_0  } \to 1 \quad \text{ in probability} \: \: \text{   and   } \: \: \frac{ FDR}{ \alpha \pi_0  } \to 1  \: \: \text{   as }  \min\{n,k\}\to\infty, $$
where $\pi_0 = q_0 / \{k(k-1)/2\}$ and $q_0$ is the number of true null hypotheses. Here $\pi_1 = 1 - \pi_0 $ is the proportion of edges in the graph representing the GGM. \\[0.05 in] 
In view of the above testing procedure, we can define the p-values of the corresponding tests as 
\begin{equation}\label{p-values}
 p_{ij} = G( | T_{ij} | ),   \quad 1 \leq i < j \leq k.
 \end{equation}
Our approach for estimating \(\pi_0\) in the aforementioned problem is to use
the Schweder--Spjøtvall estimator on the corresponding p-values, with the tuning parameter chosen according to the method of
\cite{storey2002direct, storey2003statistical}. Since the largest p-values are most likely to be uniformly distributed, \cite{schweder1982plots} suggested that a conservative estimator of $\pi_0$ is 
\begin{equation}\label{storey}
 \hat{\pi}_{0}  ( \lambda ) = \frac{ 1 }{N (1- \lambda) } \sum\limits_{1 \leq i < j \leq k } I \{ p_{ij} > \lambda  \}= \frac{ W ( \lambda )}{ N (1 - \lambda ) }, \quad 0 \leq \lambda < 1, \: N = k(k-1)/2.
\end{equation} 
The estimator in \eqref{storey} involves a tuning parameter, $\lambda$. Choosing $\lambda$ requires balancing bias and variance for the estimator $\hat{\pi}_0(\lambda)$. According to \cite{storey2003statistical}, for well-behaved p-values, the bias tends to decrease as $\lambda$ increases, reaching its minimum as $\lambda$ approaches 1. Consequently, \cite{storey2003statistical} proposed the following method for selecting the tuning parameter $\lambda$. 
\vspace{4 pt}
\setlength{\textfloatsep}{12pt}
\setlength{\floatsep}{7pt}
\setlength{\intextsep}{7pt}

\begin{algorithm}[H]
\renewcommand{\thealgorithm}{} 
\floatname{algorithm}{}        
\caption{Estimation of $\pi_0$ based on smoothing splines}
\begin{algorithmic}[1]
    \STATE Fix a set of $\lambda$ values, denoted by $\Lambda$ (e.g., $\Lambda = \{0, 0.01, 0.02, \dots, 0.95\}$).
    \STATE For each $\lambda \in \Lambda$, compute $\hat{\pi}_0(\lambda)$ as in \eqref{storey}.
    \STATE Fit a cubic spline $\hat{f}$ to the values $\hat{\pi}_0(\lambda)$.
    \STATE Obtain the final estimate:
    \[
    \hat{\pi}_0 = \min \{ \hat{f}(1), 1 \}.
    \]
\end{algorithmic}
\end{algorithm}
\noindent
In addition to the smoothing-based technique introduced by \cite{storey2003statistical}, a bootstrap-based method for selecting the optimal value of \( \lambda \) was proposed in \cite{storey2004strong}. This method builds on the earlier work of \cite{storey2002direct}. The proposed automatic choice of \( \lambda \) aims to estimate the value that minimizes the mean squared error (MSE), balancing bias and variance. Specifically, it seeks to minimize \( E \left[ (\hat{\pi}_0(\lambda) - \pi_0)^2 \right] \). The expectation is taken with respect to the distribution of $\mathbf{X}_1, \ldots, \mathbf{X}_n \overset{\text{i.i.d.}}{\sim} N(\boldsymbol{\mu}, \boldsymbol{\Sigma})$, where $\pi_0$ denotes the proportion of zero off-diagonal entries in the precision matrix. This convention applies to all subsequent expectation operators throughout the manuscript. The procedure for this method is summarized below.
\setlength{\textfloatsep}{5pt}
\setlength{\floatsep}{5pt}
\setlength{\intextsep}{5pt}

\begin{algorithm}[H]
\renewcommand{\thealgorithm}{} 
\floatname{algorithm}{}        
\caption{Bootstrap-based selection of the tuning parameter $\lambda$ and estimation of $\pi_0$}
\begin{algorithmic}[1]
    \STATE For each $\lambda \in \Lambda$, compute $\hat{\pi}_0(\lambda)$ as in \eqref{storey}.
    \vspace{0.05 in}
    \STATE Generate $B$ bootstrap samples from the p-values. For each $b = 1, \dots, B$, compute the estimators 
    \(\{\hat{\pi}_0^{*b}(\lambda)\}_{\lambda \in \Lambda}\) based on the $b$-th sample.
    \vspace{0.05 in}
    \STATE Since \(E[\hat{\pi}_0(\lambda)] \geq \pi_0\) for all \(\lambda \in [0,1)\), a plug-in estimator of \(\pi_0\) can be taken as
    \[
    \hat{\pi}_0 = \min_{\lambda' \in \Lambda} \{\hat{\pi}_0(\lambda')\}.
    \]
    \STATE For each $\lambda \in \Lambda$, estimate its respective mean squared error (MSE) as
    \[
    \widehat{\text{MSE}} (\lambda) = \frac{1}{B} \sum_{b=1}^{B} 
    \Big[ \hat{\pi}_0^{*b}(\lambda) - \min_{\lambda' \in \Lambda} \{\hat{\pi}_0(\lambda')\} \Big]^2.
    \]
    \STATE Set the final estimator as
       \[ 
       \hat{\pi}_0
=
\min\left\{1,\hat{\pi}_0(\hat{\lambda})\right\},  \quad \text{ where } \: \: 
\hat{\lambda}
\in
\operatorname*{arg\,min}_{\lambda\in\Lambda}
\widehat{\mathrm{MSE}}(\lambda).
\]
\end{algorithmic}
\end{algorithm}
\noindent
The estimator based on \eqref{storey} provides a conservative estimate of
$\pi_0$, which is an important property when the estimate is used in
data-adaptive multiple testing procedures, since underestimation of $\pi_0$ may
lead to a liberal procedure, in the sense that a multiple-testing error rate,
such as the FDR, may exceed its nominal level. In our setting, the vector of $p$-values has length $N = k(k-1)/2$. The following section outlines conditions under which the ECDF of the $p$-values converges to the average of their corresponding cumulative distribution functions (CDFs), thereby making the estimator based on \eqref{storey} an appropriate and reliable choice for our problem.

\section{Theoretical results}

We first state the conditions under which the ECDF of the
\(N=k(k-1)/2\) p-values converges to the average of their
corresponding CDFs, as
shown in Theorem~\ref{th1}. We then discuss a few examples of
covariance matrices that are particularly relevant for genetic
association studies and satisfy the weak dependency conditions of
Theorem~\ref{th1}.
\begin{theorem}
\label{th1}
Consider the test statistics in \eqref{TS} and the corresponding two-sided p-values $(p_{ij})_{1 \le i < j \le k}$ as defined in \eqref{p-values}.
Suppose $\widehat{F}_N(\cdot)$ is the ECDF of the p-values and 
\[
\bar{F}_N(x) = \frac{1}{N} \sum_{1 \le i < j \le k} F_{ij}(x),
\]
the average cdf of the p-values, where $F_{ij}(x) = \Pr(p_{ij} \le x)$. \\[0.05 in] 
Then, under Condition~\eqref{C1} and if $\log k=o(\sqrt n)$, we have for every $x\in[0,1]$,
\begin{align*}
\text{If } \sum_{i<j} |\omega_{ij}| = o(k^2), 
\quad \text{then } \widehat{F}_N(x)-\bar{F}_N(x)\rightarrow 0 \: \text{ in probability} \: \: \: 
\text{as} \: \min\{n,k\}\to\infty. 
\end{align*} 
Additionally, suppose that
\begingroup
\renewcommand{\theequation}{C2}
\renewcommand{\theHequation}{C2}
\begin{equation}\label{C2}
\frac{1}{N} \sum_{1\le i<j\le k}
\min\left\{
1,\,
\sqrt{\log k}\,
\frac{|\omega_{ij}|}{\sqrt{\omega_{ii}\omega_{jj}}}
\right\}
\longrightarrow0
\qquad\text{as } \min\{n,k\}\to\infty.
\end{equation}
\endgroup
\addtocounter{equation}{-1}
Then 
\(
\|\widehat{F}_N-\bar{F}_N\|_\infty
=
\sup_{x\in\mathbb{R}} |\widehat{F}_N(x)-\bar{F}_N(x)|
\rightarrow 0
\)
in probability as $\min\{n,k\}\to\infty$.
\end{theorem}
\noindent
The proof of this theorem is discussed in the appendix. Although the conditions in Theorem \ref{th1} are imposed on the precision matrix rather than the covariance matrix, we would like to emphasize that a broad class of covariance matrices also satisfy the weak dependency assumption of Theorem \ref{th1}. The literature on weak dependency among random variables is extensive (see, for example, \cite{MR112166}, \cite{billingsley2017probability}). In Section~2.1 of \cite{10.3150/25-BEJ1858}, several weak-dependence structures
are discussed in the context of genetic association studies. Similar dependence
structures may also arise in other high-dimensional molecular settings,
including proteomic network analyses such as the one considered in
Section~\ref{real_data}. Covariance matrices under such weakly dependent setups satisfy the assumptions of Theorem \ref{th1}. \\[0.05in]
\textbf{Example (Block-dependence structure):} Suppose the covariance matrix admits a block-diagonal structure with $b$ blocks of sizes
$k_1, \ldots, k_b$, where $\sum_{t=1}^{b} k_t = k$.
If $\max_{1 \le t \le b} k_t = o(k)$, then the total number of nonzero entries in the covariance matrix is $o(k^2)$. The block-diagonal structure is preserved under inversion. Hence, the number of nonzero entries in the precision matrix is at most
\[
\sum_{t=1}^{b}k_t^2
\le
k\max_{1\le t\le b}k_t
=
o(k^2).
\]
Condition~\eqref{C1} implies that
\(
|\omega_{ij}|
\le
\sqrt{\omega_{ii}\omega_{jj}}
\le
c_0.
\)
Therefore,
\(
\sum_{1\le i<j\le k}|\omega_{ij}|=o(k^2).
\)
Moreover,
\[
\frac{1}{N}
\sum_{1\le i<j\le k}
\min\left\{
1,\,
\sqrt{\log k}\,
\frac{|\omega_{ij}|}{\sqrt{\omega_{ii}\omega_{jj}}}
\right\}
\le
\frac{1}{N}
\sum_{1\le i<j\le k}I_{\{\omega_{ij}\ne0\}}
=
o(1).
\]
Hence, the assumptions of Theorem \ref{th1} are satisfied in this case. \\[0.03 in]
The precision matrices under a block-dependence structure introduce an interesting decomposition 
of the variables into disjoint connected components based on conditional dependence. 
This means the variables $ \mathbf{X} = (X_1, \dots, X_k)' $ can be partitioned into $l \leq k$ connected components such that 
\[
\omega_{ij} = 0 \quad \text{whenever } i \in C_t, \: j \in C_{t'} \text{ for } t \neq t'.
\] 
Consequently, their joint density factorizes as
\[
f(\mathbf{x}) = \prod_{t=1}^{l} f(\mathbf{x}_{C_t}).
\] 
\textbf{Example (Banded covariance structure):} In this case, we assume that $\mathrm{Cov}(X_i, X_j) = \sigma_{ij} = 0$ for $|i-j| > m$, where $m$ is a banding parameter. We further assume that $m = o(k)$. The banded covariance structure naturally holds under $m$-dependence. The following lemma shows that the weak-dependence assumption required to establish pointwise convergence in Theorem~\ref{th1} is satisfied under a banded covariance structure when $m=o(k)$. Moreover, under the additional condition $m\sqrt{\log k}=o(k)$, Condition~\eqref{C2} is also satisfied.
\begin{lemma}\label{l1}
Let $\boldsymbol{\Sigma} $ be a $k \times k$ banded covariance matrix with banding parameter $m$, and suppose that there exist constants \(0<c_-<c_+<\infty\) such that
\[
c_-\le\lambda_{\min}(\boldsymbol{\Sigma})
\le\lambda_{\max}(\boldsymbol{\Sigma})
\le c_+.
\]
Let $\mathbf{\Omega} = \boldsymbol{\Sigma} ^{-1} = (\omega_{ij})_{k \times k} $ denote its precision matrix. Then, the sum of the absolute values of the entries of $\mathbf{\Omega}$ satisfies
\[
\sum_{i=1}^{k} \sum_{j=1}^{k} |\omega_{ij}| = O(km).
\]
\end{lemma} 

\noindent
\textbf{ Proof :} By Theorem 2.4 of \cite{demko1984decay}, there exist constants \(C>0\) and \(0<r<1\), independent of \(k\), such that
\[
|\omega_{ij}| \le C \, r^{|i-j|/2m}, \quad \forall \; 1 \le i,j \le k,
\]
where
\[
r = \left( \frac{\sqrt{\mathrm{Cond}(\boldsymbol{\Sigma} )} - 1}{\sqrt{\mathrm{Cond}(\boldsymbol{\Sigma} )} + 1} \right)^2,
\]
and
$\mathrm{Cond}(\boldsymbol{\Sigma}) = \lambda_{\max}(\boldsymbol{\Sigma}) / \lambda_{\min}(\boldsymbol{\Sigma})$ is the condition number of $\boldsymbol{\Sigma} $.  \\[0.05 in]
It then follows that
\[
\sum_{i=1}^{k} \sum_{j=1}^{k} |\omega_{ij}| \le 2 C \sum_{t=0}^{k-1} (k-t) r_1^t, \quad \text{where } r_1 = r^{1/2m}.
\]
Since $(1-r_1)^{-1}=O(m)$,
\[
\sum_{t=0}^{k-1} (k-t) r_1^t = O(km),
\] 
and hence 
\[
\sum_{i,j=1}^{k} |\omega_{ij}| = O(km),
\] 
which establishes the result.   \\[0.05 in]
Consequently, if \(m=o(k)\), then
\(
\sum_{1\le i<j\le k}|\omega_{ij}|=o(k^2).
\)
If, in addition, \(m\sqrt{\log k}=o(k)\), then condition \eqref{C2} is also satisfied.
\hfill\(\square\)
\\[0.05 in]
For each $k$, let $I_0(k)=\{(i,j):1\leq i<j\leq k,\ \omega_{ij}=0\}$ and $I_1(k)=\{(i,j):1\leq i<j\leq k,\ \omega_{ij}\neq0\}$ denote the sets of true and false null hypotheses, respectively. Since $N=k(k-1)/2$, the corresponding proportions are $\pi_0=|I_0(k)|/N$ and $\pi_1=|I_1(k)|/N=1-\pi_0$.
\begin{corollary}[]\label{cor1}
Suppose that the assumptions of Theorem~\ref{th1}, including Condition
\eqref{C2}, hold. Then $\hat{\pi}_0(\lambda)$ is asymptotically biased
upwards for every fixed $\lambda\in[0,1)$. If $|I_1(k)|>0$, then
$$
\hat{\pi}_0(\lambda)-
\left\{
\pi_0+\pi_1\frac{\overline F_{1,n,k}(\lambda)}{1-\lambda}
\right\}
\longrightarrow0
$$
in probability as $\min\{n,k\}\to\infty$, where
$$
\overline F_{1,n,k}(\lambda)
=
\frac{1}{|I_1(k)|}
\sum_{(i,j)\in I_1(k)}
\Pr(p_{ij}>\lambda).
$$
Moreover, define
$$
F_{n,k}(\lambda)
=
\frac{1}{N}
\sum_{1\le i<j\le k}
\Pr(p_{ij}\le\lambda),
\qquad 0\le\lambda\le1.
$$
Then there exists a sequence of concave functions
$F^\circ_{n,k}$ on $[0,1]$ such that
$$
\left\|F_{n,k}-F^\circ_{n,k}\right\|_\infty
\longrightarrow0
$$
as $\min\{n,k\}\to\infty$. Consequently, every pointwise limit of
$F_{n,k}$ is concave.
\end{corollary}
\textbf{Remarks} 
\begin{enumerate}[label=(\arabic*)]
\item Near $\lambda = 1$, the Schweder-Spjøtvall estimator exhibits high variance. It is therefore common to search for an optimal $\lambda$ on a grid that does not include $1$. In R, the \texttt{smoother} method of the \texttt{qvalue} package computes the final estimator at $0.95$ by default. Hence, it is advisable to search for an optimal $\lambda$ that is bounded away from $1$. 

\item Under the concavity of the limiting alternative distribution, if we assume the existence of some $\lambda_0 < 1$ such that the alternative distribution is supported on $[0,\lambda_0]$, then the Schweder-Spjøtvall estimator evaluated at $\lambda_0$, i.e., $\hat{\pi}_0(\lambda_0)$, is a consistent estimator of $\pi_0$. 

\item  Since the average p-value distribution is asymptotically
approximated by a concave function, the Grenander estimator of \cite{langaas2005estimating}, based on the assumption of a decreasing density, can also be used to estimate \(\pi_0\). 
\end{enumerate} 

\section{Simulation studies}
We generated datasets with \( n = 200 \) observations and \( k = 100, 200, 500, 1000 \) features, for various choices of the covariance matrix \( \mathbf{ \Sigma} \). Specifically, we considered three different structures for \( \mathbf{ \Sigma} \): a block-diagonal matrix, a band graph, and an Erd\H{o}s–R\'enyi random graph, following the setup in \cite{liu2013gaussian}. \\[0.05 in]
\textbf{ Block diagonal $\Sigma $} : The choice of a block-diagonal covariance matrix is of interest because it preserves sparsity after inversion. In particular, we simulated data under a block-diagonal structure for \(\boldsymbol{\Sigma}  \), consisting of \( b \) blocks of size \( s \) (so that \( bs = k \)). Within this block-diagonal structure, the proportion of edges is given by \(
\pi_1 = (s-1)/(k-1)
\). Two types of within-block correlation structures were considered: (i) \textit{autoregressive of order 1 (AR(1))}, and (ii) \textit{equicorrelated}. 
For application of the GFC procedure, we considered two scenarios: the Lasso estimator (\(\text{GFC}_L\)) and the scaled Lasso estimator (\(\text{GFC}_{SL}\)). \\[0.05 in]
For running the simulations, we used the \texttt{SILGGM} package in \textit{R} for implementing the GFC procedure, and the \texttt{qvalue} package for evaluating the performance of Storey's estimator with the `Bootstrap' and `Smoother' methods. \\[0.05 in]
When each block of the covariance matrix follows an AR(1) covariance structure with intra-block autocorrelation parameter $0.5$, the resulting ECDF closely aligns with the uniform CDF, corresponding to the $45^{\circ}$ line through the origin. This behavior reflects the high degree of sparsity induced by the block-diagonal structure. Table~\ref{tab:AR(1)_combined_abbrev} displays the average estimated values based on 100 replications. Overall, the estimates obtained using Lasso tend to be slightly higher than those from Scaled Lasso, and in both cases, they are very close to $1$.  \\[0.03 in]
\begin{table}[t]
\centering
\caption{Simulation results under the AR(1) block covariance structure}
\label{tab:AR(1)_combined_abbrev}
\small
\setlength{\tabcolsep}{4pt} 

\begin{tabular}{c|c|cc|cc|cc|cc}
\hline
\multirow{2}{*}{$\pi_0$} &
\multirow{2}{*}{Method} &
\multicolumn{2}{c|}{$k=100$} &
\multicolumn{2}{c|}{$k=200$} &
\multicolumn{2}{c|}{$k=500$} &
\multicolumn{2}{c}{$k=1000$} \\
\cline{3-10}
& & Smoother & Bootstrap 
  & Smoother & Bootstrap 
  & Smoother & Bootstrap 
  & Smoother & Bootstrap \\
\hline

\multirow{2}{*}{0.80}
& GFC$_L$    & 0.97 & 0.96 & 0.99 & 0.98 & 0.98 & 0.98 & 0.99 & 0.99 \\
& GFC$_{SL}$ & 0.96 & 0.95 & 0.99 & 0.98 & 0.98 & 0.98 & 0.99 & 0.99 \\[2pt]

\multirow{2}{*}{0.90}
& GFC$_L$    & 0.98 & 0.97 & 0.98 & 0.98 & 0.99 & 0.99 & 0.99 & 0.99 \\
& GFC$_{SL}$ & 0.98 & 0.97 & 0.98 & 0.98 & 0.99 & 0.99 & 0.99 & 0.99 \\[2pt]

\multirow{2}{*}{0.95}
& GFC$_L$    & 0.95 & 0.96 & 0.98 & 0.98 & 0.99 & 0.99 & 0.99 & 0.99 \\
& GFC$_{SL}$ & 0.95 & 0.96 & 0.98 & 0.98 & 0.99 & 0.99 & 0.99 & 0.99 \\

\hline
\end{tabular}

\vspace{1mm}
\footnotesize
\emph{Note:} GFC$_L$ represents the GFC procedure with Lasso-based optimization, and GFC$_{SL}$ represents the GFC procedure with scaled Lasso.

\end{table}
\noindent
In Figure~\ref{ecdf_bd_2}, we show the ECDF of the p-values for data generated from a block-diagonal covariance matrix $\boldsymbol{\Sigma} $, where each block exhibits equicorrelation with correlation $0.5$. Such equicorrelated blocks introduce stronger dependence within the model, resulting in concave p-value distributions in both scenarios. 
\begin{figure}[t]
\centering
\subfloat[ECDF of p-values based on $\mathrm{GFC}_L$ (GFC procedure with Lasso-based optimization).]{%
\resizebox*{7.5cm}{!}{\includegraphics{ 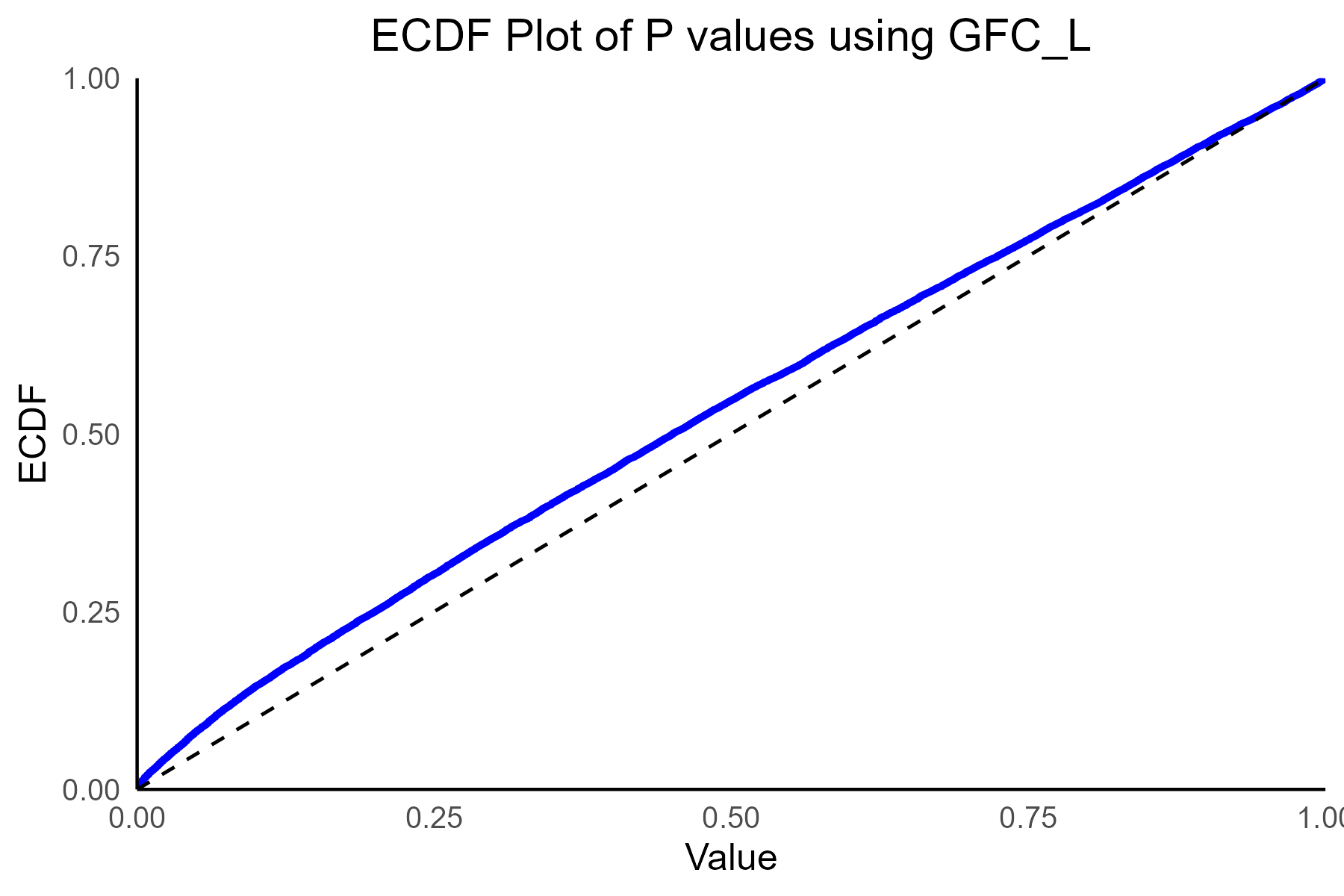 }}}\hspace{5pt}
\subfloat[ ECDF of p-values based on $\mathrm{GFC}_{SL}$ (GFC procedure with scaled Lasso).]{%
\resizebox*{7.5cm}{!}{\includegraphics{ 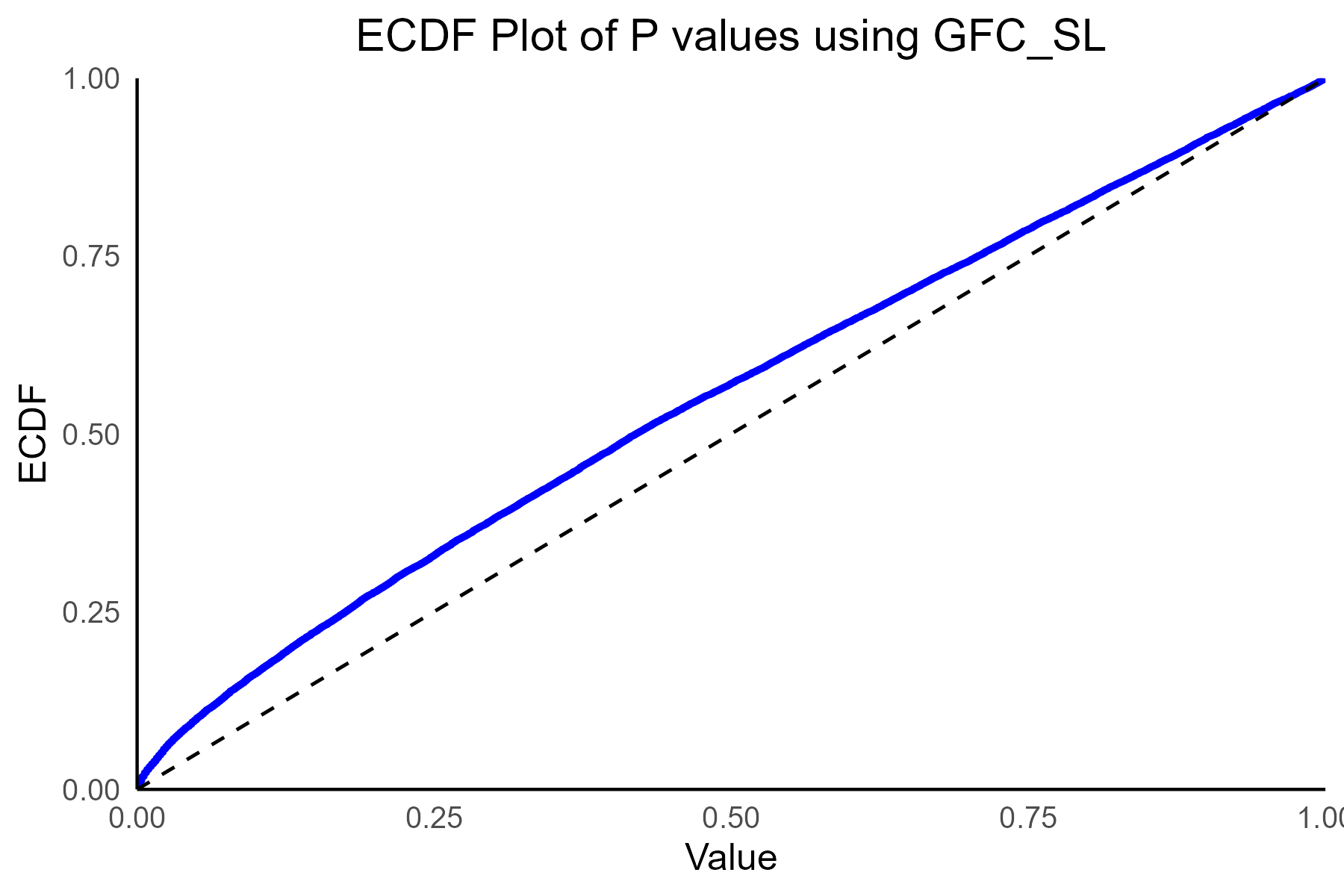 }}}
\caption{ECDFs of the $p$-values under a block-diagonal covariance structure with equicorrelated blocks $(\rho=0.5)$ $(n=200,\ k=500)$.}
\label{ecdf_bd_2}
\end{figure}
\noindent
Table~\ref{tab:equicorr_combined_abbrev} reports the average estimated values of $\pi_0$. Consistent with previous observations, Lasso-based estimates tend to be slightly higher than those obtained using Scaled Lasso. Moreover, both the ``smoother'' and ``bootstrap'' approaches continue to produce conservative estimates of $\pi_0$ in this setting. \\[0.03 in]
\begin{table}[t]
\centering
\caption{Simulation results under the equicorrelated ($\rho=0.5$) block covariance structure}
\label{tab:equicorr_combined_abbrev}
\small
\setlength{\tabcolsep}{4pt} 

\begin{tabular}{c|c|cc|cc|cc|cc}
\hline
\multirow{2}{*}{$\pi_0$} &
\multirow{2}{*}{Method} &
\multicolumn{2}{c|}{$k=100$} &
\multicolumn{2}{c|}{$k=200$} &
\multicolumn{2}{c|}{$k=500$} &
\multicolumn{2}{c}{$k=1000$} \\
\cline{3-10}
& & Smoother & Bootstrap 
  & Smoother & Bootstrap 
  & Smoother & Bootstrap 
  & Smoother & Bootstrap \\
\hline

\multirow{2}{*}{0.80} 
& GFC$_L$    & 0.93 & 0.93 & 0.97 & 0.96 & 0.98 & 0.98 & 0.98 & 0.98 \\
& GFC$_{SL}$ & 0.91 & 0.91 & 0.94 & 0.94 & 0.95 & 0.95 & 0.94 & 0.94 \\[2pt]

\multirow{2}{*}{0.90} 
& GFC$_L$    & 0.93 & 0.94 & 0.96 & 0.96 & 0.98 & 0.98 & 0.98 & 0.98 \\
& GFC$_{SL}$ & 0.93 & 0.91 & 0.92 & 0.93 & 0.94 & 0.94 & 0.94 & 0.94 \\[2pt]

\multirow{2}{*}{0.95} 
& GFC$_L$    & 0.94 & 0.93 & 0.97 & 0.97 & 0.98 & 0.98 & 0.99 & 0.99 \\
& GFC$_{SL}$ & 0.92 & 0.91 & 0.93 & 0.93 & 0.95 & 0.95 & 0.94 & 0.94 \\

\hline
\end{tabular}

\vspace{1mm}
\footnotesize
\emph{Note:} GFC$_L$ represents the GFC procedure with Lasso-based optimization, and GFC$_{SL}$ represents the GFC procedure with scaled Lasso.
\end{table}
\noindent
\textbf{Band Graph :} We consider the same band graph as in \cite{liu2013gaussian}, where
\(\Omega = (\omega_{ij})\) satisfies
\[
\omega_{ij} =
\begin{cases} 
1, & \text{if } i = j, \\[2mm]
0.6, & \text{if } |i - j| = 1, \\[2mm]
0.3, & \text{if } |i - j| = 2, \\[2mm]
0, & \text{if } |i - j| \geq 3.
\end{cases}
\]
For this band graph, out of the \(k(k-1)\) off-diagonal elements, only
\(2(2k-3)\) elements are nonzero.  
Hence, the proportion of nonzero off-diagonal elements is
\(
\pi_1 = (4k-6)/\{ k(k-1)\}.
\) 
As $k$ increases, the precision matrix becomes increasingly sparse, and the estimated $\hat{\pi}_0$ values for both methods rise accordingly, approaching $1$ as expected. Table \ref{tab:Band_mat} displays the average estimated values based on $100$ replications.  \\[0.05 in]
\begin{table}[t]
\centering
\small
\setlength{\tabcolsep}{4pt} 
\caption{Estimated $\hat{\pi}_0$ for the band graph precision matrix ($\pi_0 = 1- (4k-6)/\{ k(k-1) \}$)}
\label{tab:Band_mat}

\begin{tabular}{c|cc|cc|cc|cc}
\hline
\multirow{2}{*}{Method} &
\multicolumn{2}{c|}{\shortstack{$k=100$\\$\pi_0=0.96$}} &
\multicolumn{2}{c|}{\shortstack{$k=200$\\$\pi_0=0.98$}} &
\multicolumn{2}{c|}{\shortstack{$k=500$\\$\pi_0=0.99$}} &
\multicolumn{2}{c}{\shortstack{$k=1000$\\$\pi_0=1.00$}} \\
\cline{2-9}
& Smoother & Bootstrap & Smoother & Bootstrap & Smoother & Bootstrap & Smoother & Bootstrap \\ 
\hline
GFC$_L$  & 0.95 & 0.94 & 0.97 & 0.97 & 0.98 & 0.98 & 0.99 & 0.99 \\
GFC$_{SL}$ & 0.96 & 0.95 & 0.97 & 0.97 & 0.98 & 0.98 & 0.99 & 0.99 \\[2pt]
\hline
\end{tabular}

\vspace{1mm}
\footnotesize
\emph{Note:} GFC$_L$ represents the GFC procedure with Lasso-based optimization, and GFC$_{SL}$ represents the GFC procedure with scaled Lasso.
\end{table}
\noindent
\textbf{Erd\H{o}s–R\'enyi random graph : } In this case, $\omega_{ij} = u_{ ij} \delta_{ij} $ where $u_{ij} \sim Unif \: [ 0.4 , 0.8 ]$ and $ \delta_{ij} \sim Ber ( q ) $  where $ q = \min \{ 0.05 , 5/k \} $. Table \ref{tab:ERRG} provides the $\hat{\pi}_0$ under different combinations.
\begin{table}[t]
\centering
\small
\setlength{\tabcolsep}{4pt} 
\caption{Estimated $\hat{\pi}_0$ for the Erd\H{o}s–R\'enyi random graph ($\pi_0 = 1 - \min \{ 0.05 , 5/k \} $)}
\label{tab:ERRG}

\begin{tabular}{c|cc|cc|cc|cc}
\hline
\multirow{2}{*}{Method} &
\multicolumn{2}{c|}{\shortstack{$k=100$\\$\pi_0=0.95$}} &
\multicolumn{2}{c|}{\shortstack{$k=200$\\$\pi_0=0.98$}} &
\multicolumn{2}{c|}{\shortstack{$k=500$\\$\pi_0=0.99$}} &
\multicolumn{2}{c}{\shortstack{$k=1000$\\$\pi_0=1.00$}} \\
\cline{2-9}
& Smoother & Bootstrap & Smoother & Bootstrap & Smoother & Bootstrap & Smoother & Bootstrap \\ 
\hline
GFC$_L$   & 0.94 & 0.93 & 0.97 & 0.96 & 0.98 & 0.98 & 0.99 & 0.99 \\
GFC$_{SL}$ & 0.95 & 0.93 & 0.97 & 0.96 & 0.98 & 0.98 & 0.98 & 0.98 \\[2pt]
\hline
\end{tabular}

\vspace{1mm}
\footnotesize
\emph{Note:} GFC$_L$ represents the GFC procedure with Lasso-based optimization, and GFC$_{SL}$ represents the GFC procedure with scaled Lasso.
\end{table}
\noindent
We also accommodate the case where the sparsity of the Erd\H{o}s--R\'enyi random graph 
does not change as $k$ increases. This means that $\delta_{ij} \sim \mathrm{Ber}(q)$ 
for some fixed $0 < q < 1$. We consider $q = 0.2, 0.1,$ and $0.05$ 
(i.e., $\pi_0 = 0.8, 0.9,$ and $0.95$). In this setting, where the value of $\pi_1$ is slightly larger, the ECDF of the $p$-values shows a modest deviation from the uniform $[0,1]$ distribution. Nevertheless, the concave shape of the ECDF suggests that Storey’s estimator, with an appropriately chosen tuning parameter, can still yield a reasonable estimate of $\pi_0$. For illustration, we highlight the ECDF corresponding to $q = 0.2$, $n = 200$, and $k = 500$ in figure \ref{ecdf_ERRG_1}. 
Table~\ref{tab:ERRG_Fixed} presents the estimated $\hat{\pi}_0$ values 
for these different combinations.
\begin{table}[t]
\centering
\small
\setlength{\tabcolsep}{4pt} 
\caption{Estimated $\hat{\pi}_0$ for the Erd\H{o}s–R\'enyi random graph with fixed sparsity}
\label{tab:ERRG_Fixed}

\begin{tabular}{c|c|cc|cc|cc|cc}
\hline
\multirow{2}{*}{$\pi_0$} &
\multirow{2}{*}{Method} &
\multicolumn{2}{c|}{$k=100$} &
\multicolumn{2}{c|}{$k=200$} &
\multicolumn{2}{c|}{$k=500$} &
\multicolumn{2}{c}{$k=1000$} \\
\cline{3-10}
& & Smoother & Bootstrap & Smoother & Bootstrap & Smoother & Bootstrap & Smoother & Bootstrap \\ 
\hline
\multirow{2}{*}{0.80} 
& GFC$_L$  & 0.83 & 0.83 & 0.87 & 0.87 & 0.93 & 0.93 & 0.95 & 0.95 \\
& GFC$_{SL}$ & 0.80 & 0.81 & 0.82 & 0.83 & 0.86 & 0.86 & 0.88 & 0.88 \\[2pt]

\multirow{2}{*}{0.90} 
& GFC$_L$  & 0.90 & 0.90 & 0.90 & 0.91 & 0.94 & 0.94 & 0.96 & 0.96 \\
& GFC$_{SL}$ & 0.89 & 0.89 & 0.89 & 0.89 & 0.89 & 0.89 & 0.89 & 0.89 \\[2pt]

\multirow{2}{*}{0.95} 
& GFC$_L$  & 0.95 & 0.94 & 0.95 & 0.94 & 0.96 & 0.96 & 0.96 & 0.97 \\
& GFC$_{SL}$ & 0.94 & 0.94 & 0.93 & 0.93 & 0.91 & 0.92 & 0.91 & 0.91 \\[2pt]
\hline
\end{tabular}

\vspace{1mm}
\footnotesize
\emph{Note:} GFC$_L$ represents the GFC procedure with Lasso-based optimization, and GFC$_{SL}$ represents the GFC procedure with scaled Lasso.
\end{table}
\begin{figure}[t]
\centering
\subfloat[ECDF of p-values based on $\mathrm{GFC}_L$ (GFC procedure with Lasso-based optimization).]{%
\resizebox*{7.5cm}{!}{\includegraphics{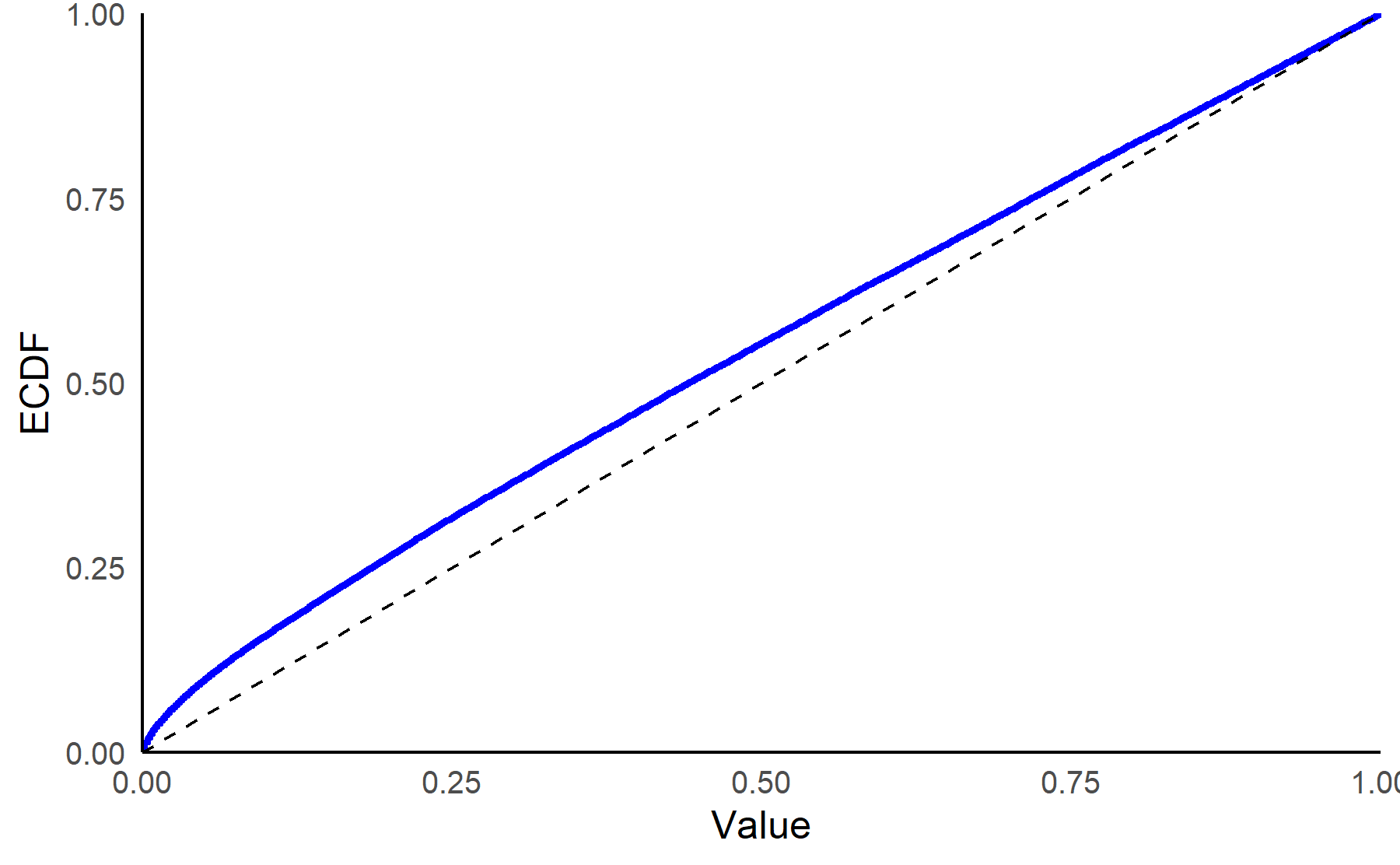 }}}\hspace{5pt}
\subfloat[ECDF of p-values based on $\mathrm{GFC}_{SL}$ (GFC procedure with scaled Lasso).]{%
\resizebox*{7.5cm}{!}{\includegraphics{ 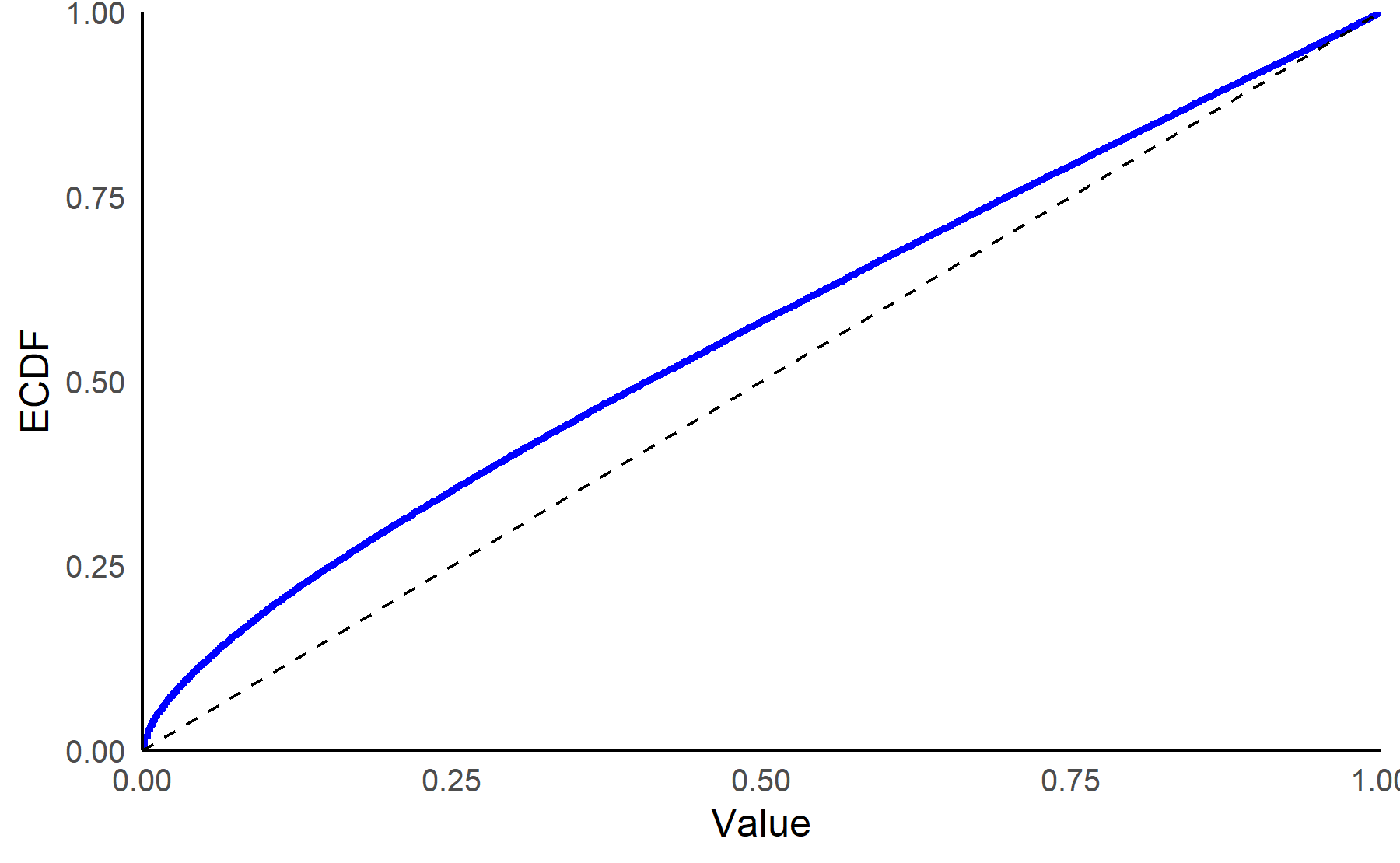 }}}
\caption{ECDF of p-values for an Erd\H{o}s--R\'enyi random graph with fixed sparsity ($q=0.2$, $n=200$, $k=500$).} \label{ecdf_ERRG_1}
\end{figure}

\noindent
In all cases, the general observation is that the estimated $\hat{\pi}_0$ is close to the true value of $\pi_0$. Thus, the GFC procedure, combined with Storey's estimator, provides a reasonable estimate of the graph's complexity. However, for $\pi_0 = 0.95$ and $k = 1000$, this estimator slightly underestimates $\pi_0$. It should be noted that this scenario deviates from our sparsity assumptions of Theorem \ref{th1}. Interestingly, even with this slight violation of the model assumptions, the method still provides a reasonable estimate of the true graph complexity.
{
\section{Real data analysis}\label{real_data}
To illustrate the practical performance of the proposed method, we consider the
breast cancer proteomic dataset from The Cancer Genome Atlas (TCGA) study of
\cite{cancer2012comprehensive}. The dataset contains measurements of $k=171$
cancer-related proteins and phosphoproteins obtained from $n=403$ primary
breast tumor samples. The measurements
were generated using reverse-phase protein arrays (RPPA), a platform that
measures the relative abundance of selected proteins and modified forms of
proteins in the tumor tissue. Thus, each row of the data matrix corresponds to
a primary breast tumor sample from a breast cancer patient, while each column
corresponds to a protein or phosphoprotein measurement. The dataset is publicly
available through the National Cancer Institute's Genomic Data Commons (GDC).
Related TCGA breast cancer RPPA data have previously been analyzed using
graphical Lasso methods to estimate the underlying Gaussian graphical model;
see \cite{lingjaerde2021tailored}. Our objective is different. Rather than
estimating the individual edges, we consider the $N=k(k-1)/2$ hypotheses
concerning the off-diagonal precision-matrix entries and estimate the overall
network complexity from the resulting p-values.
\\[0.05in]
For this dataset, $n=403$ and $k=171$, so that the multiple testing problem
involves $14{,}535$ hypotheses. Table~\ref{tab:TCGA_RPPA} presents the estimates
of the proportion of null hypotheses, $\hat{\pi}_0$, obtained using the
smoother and bootstrap methods for both $\mathrm{GFC}_{L}$ and
$\mathrm{GFC}_{SL}$.
\begin{table}[H]
\centering
\small
\setlength{\tabcolsep}{6pt}
\caption{Estimated sparsity levels for the TCGA breast cancer
RPPA data.}
\label{tab:TCGA_RPPA}
\begin{minipage}{0.72\textwidth}
\centering
\begin{tabular}{|ll|ll|}
\hline
\multicolumn{2}{|c|}{$\mathrm{GFC}_{L}$} &
\multicolumn{2}{c|}{$\mathrm{GFC}_{SL}$} \\ \hline
Smoother & Bootstrap & Smoother & Bootstrap \\ \hline
0.76 & 0.76 & 0.73 & 0.73 \\ \hline
\end{tabular}
\par\smallskip
\raggedright
\footnotesize
\emph{Note:} $\mathrm{GFC}_{L}$ denotes the GFC procedure with Lasso-based optimization, and $\mathrm{GFC}_{SL}$ denotes the GFC procedure with scaled Lasso.
\end{minipage}
\end{table}
\noindent
The estimated proportions of null hypotheses range from $0.73$ to $0.76$,
corresponding to estimated edge proportions between $0.24$ and $0.27$, which
represent the estimated complexity of the proteomic network in primary breast
tumors from breast cancer patients. For
each GFC implementation, the smoother and bootstrap methods produce the same
estimate to two decimal places, indicating that the estimated edge proportion
is not substantially affected by the choice between these two methods.
Therefore, approximately one-quarter of the $14{,}535$ possible protein and
phosphoprotein pairs are estimated to correspond to nonzero off-diagonal
precision-matrix entries. This indicates a substantial conditional-dependence
structure among the measured proteins and phosphoproteins across the primary
breast tumor samples.
\\[0.05 in]
Figure~\ref{fig:P_TCGA_RPPA} presents the ECDFs of
the p-values obtained from the two procedures. Both plots exhibit a concave
pattern, with a relatively large concentration of small p-values. This behavior
is consistent with the presence of a non-negligible proportion of non-null
hypotheses. The broadly similar conclusions obtained from
$\mathrm{GFC}_{L}$ and $\mathrm{GFC}_{SL}$ further indicate that the estimated
network complexity is reasonably stable across the two implementations.
\begin{figure}[h]
\centering
\subfloat[Empirical CDF of p-values based on $\mathrm{GFC}_{L}$ (GFC with Lasso) for the TCGA breast cancer RPPA data.]{%
\resizebox*{7.5cm}{!}{%
\includegraphics{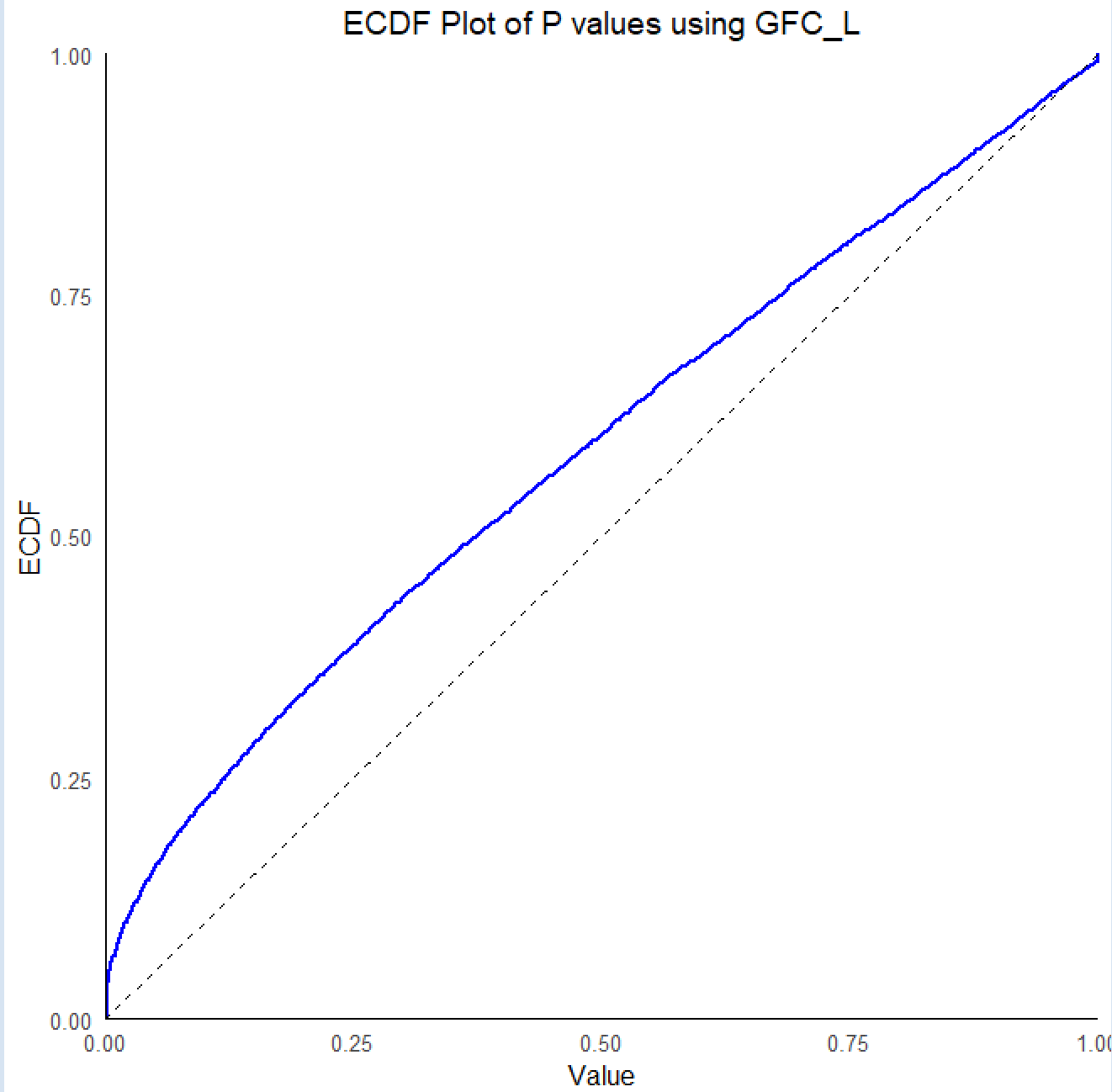}%
}}%
\hspace{5pt}
\subfloat[Empirical CDF of p-values based on $\mathrm{GFC}_{SL}$ (GFC with scaled Lasso) for the TCGA breast cancer RPPA data.]{%
\resizebox*{7.5cm}{!}{%
\includegraphics{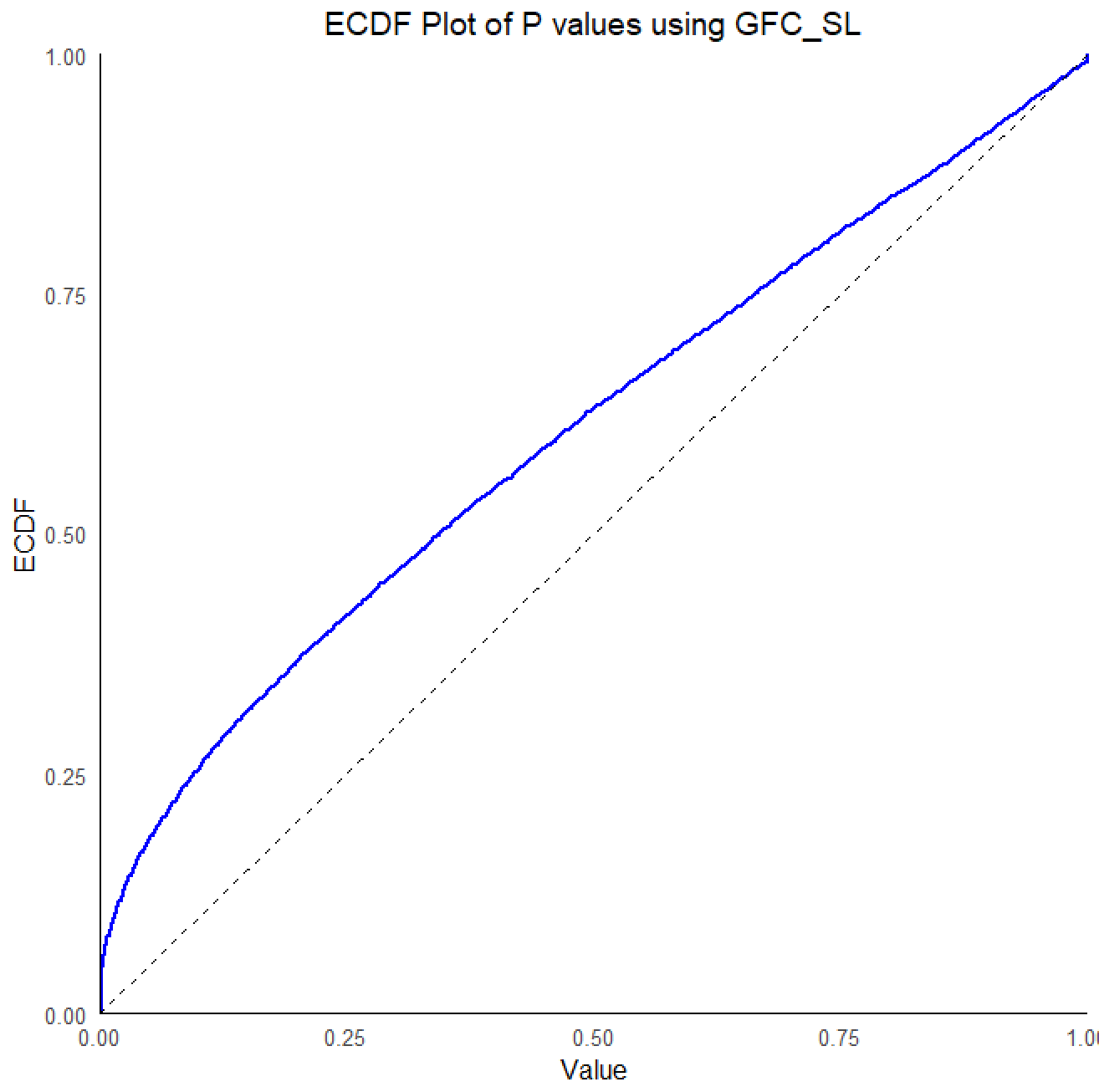}%
}}
\caption{Empirical CDFs of p-values obtained from the
$\mathrm{GFC}_{L}$ and $\mathrm{GFC}_{SL}$ procedures for the TCGA breast
cancer RPPA data of \cite{cancer2012comprehensive}.}
\label{fig:P_TCGA_RPPA}
\end{figure}
}
\section{ Concluding Remarks} 
In this paper, we address the problem of simultaneously testing the entries of a precision matrix using the methodology of \cite{liu2013gaussian}, and we show that the Schweder--Spjøtvall estimator, with the tuning parameter selected via bootstrap or smoothing splines, provides a reliable means of estimating the complexity of the underlying graph. The assumptions of Theorem~\ref{th1} are fairly general, covering a wide range of dependent structures commonly encountered in applications, particularly in gene-based studies. \\[0.05 in]
The theoretical analysis in this paper concerns the ECDF of the GFC-based $p$-values for estimating this global measure of network complexity. Thus, we do not rederive the ratio-convergence arguments used in \cite{liu2013gaussian} for FDP or FDR control at data-driven rejection thresholds. A joint theory connecting global complexity estimation with threshold-based graph recovery is an interesting direction for future work.\\[0.05 in]
Due to its ability to accommodate various weakly dependent structures and its simple formulation based on the ECDF of the $p$-values, the Schweder--Spjøtvall estimator is well-suited for problems with inherent dependencies among $p$-values. The main theoretical result of Theorem~\ref{th1} relies on assumptions regarding the sum of absolute values of the precision matrix entries. Covariance matrices with Toeplitz structure and fast decay rates (e.g., exponential decay) satisfy these assumptions. An interesting direction for future work would be to impose a similar criterion directly on the covariance matrix. Although the asymptotic small order of the sum of absolute covariances alone is not sufficient to guarantee convergence of the ECDF to the average CDF, practical examples and simulation results suggest that this criterion may hold for a broad class of covariance matrices. \\[0.05 in]
Another promising direction is to extend the methodology to copula-based graphical models instead of Gaussian graphical models, as considered in \cite{bauer2012pair}. Previous works such as \cite{dobra2011copula,liu2012high} have studied semiparametric Gaussian copula models, while multi-attribute Gaussian graphical models have been explored in \cite{li2025novel}. Latent variable-based approaches have been investigated in \cite{behrouzi2019detecting,hermes2024copula,6288344}. Moreover, \cite{neumann2021estimating} considered estimation of $\pi_0$ under a general copula model using a combination of the independent component bootstrap of \cite{hall2009using} and the Schweder-Spjøtvall estimator. Extending the estimation of network complexity to more general copula-based models represents an exciting and challenging direction for future research.













\section{Appendix}

We begin by stating a few technical lemmas that are essential
for the proof of Theorem \ref{th1}. \\[0.05 in]
Define
\begin{equation}\label{uij}
U_{ij} \;=\; \frac{1}{\sqrt{n}} \sum_{l=1}^{n}
\left\{ \varepsilon_{li} \varepsilon_{lj} - E[\varepsilon_{li}\varepsilon_{lj}] \right\}.
\end{equation}
\noindent
Observe that
\[
\begin{pmatrix}
 \varepsilon_{li} \\[3pt]
 \varepsilon_{lj}
\end{pmatrix}
\sim
N_2\!\left(
\begin{pmatrix} 0 \\ 0 \end{pmatrix},
\begin{bmatrix}
 \delta_{ii} & \delta_{ij} \\
 \delta_{ij} & \delta_{jj}
\end{bmatrix}
\right),
\]
where
\[
\delta_{ij} = \frac{\omega_{ij}}{\omega_{ii}\,\omega_{jj}},
\qquad
\delta_{ii} = \sigma_{ii}
 - \Sigma_{i,-i}\,\Sigma_{-i,-i}^{-1}\,\Sigma_{-i,i}.
\]

\begin{lemma}\label{l_cov}
For any $(i,j)$ and $(i',j')$ with $i<j$ and $i'<j'$,
\[
\operatorname{Cov}(U_{ij}, U_{i'j'})
=
\delta_{ii'} \delta_{jj'}
\;+\;
\delta_{ij'} \delta_{i'j}.
\]
\end{lemma}
\noindent\textbf{Proof :}
By Isserlis's theorem (see \cite{isserlis1918formula}),
\[
\begin{aligned}
E \!\left[
\varepsilon_{li}\varepsilon_{lj}\varepsilon_{li'}\varepsilon_{lj'}
\right]
&=
E [\varepsilon_{li}\varepsilon_{lj}]\,E[\varepsilon_{li'}\varepsilon_{lj'}]
+
E [\varepsilon_{li}\varepsilon_{li'}]\,E[\varepsilon_{lj}\varepsilon_{lj'}]
+
E [\varepsilon_{li}\varepsilon_{lj'}]\,E[\varepsilon_{li'}\varepsilon_{lj}].
\end{aligned}
\]
From this decomposition, the desired expression for
$\operatorname{Cov}(U_{ij}, U_{i'j'})$ follows immediately.
\hfill\(\square\)
\begin{lemma}\label{l_ecdf}
Let \((X,Y)\) be a standard bivariate Gaussian vector with correlation
coefficient \(\rho\). Then, for every \(x,y\in\mathbb{R}\),
\[
\left|
\operatorname{Cov}\Big(
I_{\{X\le x\}},
I_{\{Y\le y\}}
\Big)
\right|
\le
\frac{|\rho|}{4}.
\]
\end{lemma}

\noindent\textbf{Proof :}
By the maximal-correlation inequality for a bivariate Gaussian vector
\citep{lancaster1957},
\[
\left|
\operatorname{Corr}\{f(X),g(Y)\}
\right|
\le
|\rho|
\]
for all square-integrable functions \(f\) and \(g\). Taking
\[
f(X)=I_{\{X\le x\}},
\qquad
g(Y)=I_{\{Y\le y\}},
\]
we obtain
\[
\left|
\operatorname{Cov}\Big(
I_{\{X\le x\}},
I_{\{Y\le y\}}
\Big)
\right|
\le
|\rho|
\sqrt{
\operatorname{Var}\big(I_{\{X\le x\}}\big)
\operatorname{Var}\big(I_{\{Y\le y\}}\big)
}.
\]
Since the variance of an indicator function is at most \(1/4\), the
result follows. 
\hfill\(\square\)
\begin{lemma}\label{pij}
Let $\{X_i\}_{i\ge1}$ be real-valued random variables with ECDF
\[
\widehat{F}_m(x) = \frac{1}{m}\sum_{i=1}^m \mathbf{1}_{\{X_i \le x\}},
\]
and let $\bar{F}_m(x) = \frac{1}{m}\sum\limits_{i=1}^{m} F_i(x)$, where $F_i$ is the distribution function of $X_i$. Assume that
\[
\sup_{x\in\mathbb{R}} |\widehat{F}_m(x)-\bar{F}_m(x)| \to 0 \quad \text{ in probability} \quad \text{ as } \:  m \to \infty .
\]
Let $h:\mathbb{R}\to\mathbb{R}$ be continuous, and define $A(y) = h^{-1}((-\infty,y])$ for each $y$. Assume that for every $y$, $A(y)$ is a finite union of (possibly unbounded) closed intervals,
\[
A(y) = \bigcup_{j=1}^{k(y)} [a_j(y),b_j(y)], \qquad k(y)\le K<\infty,
\]
and that $\bar{F}_m$ is continuous at each boundary point $a_j(y)$ and $b_j(y)$. Define $X_i' = h(X_i)$ and
\[
\widehat{G}_m(y) = \frac{1}{m}\sum_{i=1}^m \mathbf{1}_{\{X_i' \le y\}}, \qquad
\bar{G}_m(y) = \frac{1}{m}\sum_{i=1}^m G_i(y),
\]
where $G_i$ is the distribution function of $X_i'$. Then
\[
\sup_{y\in\mathbb{R}} |\widehat{G}_m(y)-\bar{G}_m(y)| \to 0 \quad \text{ in probability} \quad \text{ as } \: m \to \infty.
\]
\end{lemma}
\noindent\textbf{Proof :}
Fix \(y\in\mathbb{R}\). After merging overlapping intervals, write
\[
A(y)=\bigcup_{j=1}^{r(y)}[a_j(y),b_j(y)]
\]
as a disjoint union, where \(r(y)\le k(y)\le K\). Then
\[
\widehat{G}_m(y)
=
\sum_{j=1}^{r(y)}
\bigl\{
\widehat{F}_m(b_j(y))-\widehat{F}_m(a_j(y)-)
\bigr\},
\]
whereas, by the continuity of \(\bar{F}_m \) at the boundary points,
\[
\bar{G}_m(y)
=
\sum_{j=1}^{r(y)}
\bigl\{
\bar{F}_m(b_j(y))-\bar{F}_m(a_j(y))
\bigr\}.
\]
Hence,
\[
\begin{aligned}
|\widehat{G}_m(y)-\bar{G}_m(y)|
&\le
\sum_{j=1}^{r(y)}
\left|
\widehat{F}_m(b_j(y))-\bar{F}_m(b_j(y))
\right| \\
&\quad+
\sum_{j=1}^{r(y)}
\left|
\widehat{F}_m(a_j(y)-)-\bar{F}_m(a_j(y))
\right| \\
&\le
2r(y)\sup_{x\in\mathbb{R}}
|\widehat{F}_m(x)-\bar{F}_m(x)| \\
&\le
2K\sup_{x\in\mathbb{R}}
|\widehat{F}_m(x)-\bar{F}_m(x)|.
\end{aligned}
\]
Taking the supremum over \(y\in\mathbb{R}\) and using the assumption
gives
\[
\sup_{y\in\mathbb{R}}|\widehat{G}_m(y)-\bar{G}_m(y)|
\to 0 \quad \text{ in probability}.
\]
\hfill\(\square\)
\begin{corollary}\label{l_pval}
Let \(T_1,\ldots,T_m\) be real-valued test statistics with ECDF
\[
\widehat{F}_m(t)=\frac{1}{m}\sum_{i=1}^m\mathbf{1}_{\{T_i\le t\}},
\]
and average CDF
\[
\bar{F}_m(t)=\frac{1}{m}\sum_{i=1}^m F_i(t),
\]
where \(F_i\) is the distribution function of \(T_i\). Assume that
\[
\sup_{t \in\mathbb{R}} |\widehat{F}_m(t)-\bar{F}_m(t)| \to  0 \quad \text{ in probability}.
\]
Define two-sided normal p-values
\[
P_i=2\bigl(1-\Phi(|T_i|)\bigr),
\]
and let
\[
\widehat{G}_m(p)=\frac{1}{m}\sum_{i=1}^m\mathbf{1}_{\{P_i\le p\}},
\qquad
\bar{G}_m(p)=\frac{1}{m}\sum_{i=1}^m G_i(p),
\]
where \(G_i\) is the distribution function of \(P_i\). Then
\[
\sup_{p \in [0,1]} |\widehat{G}_m(p)-\bar{G}_m(p)| \to 0  \quad \text{ in probability}.
\]
\end{corollary}
\noindent\textbf{Proof :}
Let \(h(t)=2\{1-\Phi(|t|)\}\), which is continuous. For
\(p\in(0,1)\), put \(c_p=\Phi^{-1}(1-p/2)\). Then
\[
P_i\le p
\quad\Longleftrightarrow\quad
|T_i|\ge c_p,
\]
and hence
\[
h^{-1}((-\infty,p])
=
(-\infty,-c_p]\cup[c_p,\infty).
\]
Thus, for \(p\in(0,1)\), the preimage is a union of two closed
intervals. For \(p=0\), the preimage is empty, whereas for \(p=1\),
it is \(\mathbb{R}\). Therefore, the assumptions of Lemma
\ref{pij} hold with \(K=2\). Consequently,
\[
\sup_{p \in [0,1]} |\widehat{G}_m(p)-\bar{G}_m(p)| \to 0 \quad \text{ in probability}.
\]
\hfill\(\square\)
\\[0.05 in]
\noindent\textbf{Proof of Theorem \ref{th1} :}
Let \(\mathcal I=\{(i,j):1\le i<j\le k\}\),
\(\theta_{ij}=\omega_{ij}/\sqrt{\omega_{ii}\omega_{jj}}\), and
\(\mu_{ij}=\sqrt n\,\theta_{ij}\), and define
\(Z_{ij}=U_{ij}/\sqrt{\delta_{ii}\delta_{jj}}\) and
\(\widetilde T_{ij}=Z_{ij}-\mu_{ij}\). Let \(\widehat{\widetilde{H}}_N\) and
\(\overline{\widetilde{H}}_{N}\) be the ECDF and the average distribution
function of \(\{\widetilde T_{ij}:(i,j)\in\mathcal I\}\), respectively.
Define \(\widehat{H}_N\) and \(\bar{H}_N\) analogously for
\(\{T_{ij}:(i,j)\in\mathcal I\}\). \\[0.05 in]
The uniform expansion in equation~(26) of
\citet{liu2013gaussian}, together with the bounds for
$\hat r_{ii}$ and $b_{n,ij}$ established in the proof of
Theorem~3.1 therein, gives
\[
\begin{aligned}
R_{n,k}
&:=
\max_{(i,j)\in\mathcal I}
\left|
Z_{ij}-T_{ij}
-\sqrt{\frac{n}{\hat r_{ii}\hat r_{jj}}}\,
b_{n,ij}\frac{\omega_{ij}}{\omega_{ii}\omega_{jj}}
\right|
=o_P(1), \\
D_{n,k}
&:=
\max_{(i,j)\in\mathcal I}
\left|
b_{n,ij}
\sqrt{\frac{\delta_{ii}\delta_{jj}}{\hat r_{ii}\hat r_{jj}}}
-1
\right|
=
O_P\left(\sqrt{\frac{\log k}{n}}\right).
\end{aligned}
\]
Since
\[
\sqrt{\frac{n}{\hat r_{ii}\hat r_{jj}}}\,
b_{n,ij}\frac{\omega_{ij}}{\omega_{ii}\omega_{jj}}
=
\mu_{ij}
b_{n,ij}
\sqrt{\frac{\delta_{ii}\delta_{jj}}{\hat r_{ii}\hat r_{jj}}},
\]
we have
\begin{equation}
|T_{ij}-\widetilde T_{ij}|
\le
R_{n,k}+D_{n,k}|\mu_{ij}|.
\label{eq:T-tildeT}
\end{equation}
For \(a=(i,j)\) and \(b=(i',j')\) in \(\mathcal I\), let
\(\rho_{a,b}=\operatorname{Corr}(U_{ij},U_{i'j'})\). Fix
\(\eta\in(0,1)\). For \(a\ne b\) with \(|\rho_{a,b}|\le\eta\), the
bivariate Berry--Esseen theorem and Lemma~\ref{l_ecdf} yield, for some
constant \(C_{1,\eta}>0\), uniformly in \(s,t\in\mathbb R\),
\[
\left|
\operatorname{Cov}\bigl(
I_{\{\widetilde T_{ij}\le s\}},
I_{\{\widetilde T_{i'j'}\le t\}}
\bigr)
\right|
\le
\frac{|\rho_{a,b}|}{4}+\frac{C_{1,\eta}}{\sqrt n}.
\]
By Lemma~\ref{l_cov},
\[
\rho_{(ij),(i'j')}
=
\frac{
\delta_{ii'}\delta_{jj'}
+\delta_{ij'}\delta_{i'j}
}{
\sqrt{
(\delta_{ii}\delta_{jj}+\delta_{ij}^{2})
(\delta_{i'i'}\delta_{j'j'}+\delta_{i'j'}^{2})
}
}.
\]
Condition~\eqref{C1} implies that, for some constant \(C_2>0\),
\[
|\rho_{(ij),(i'j')}|
\le
C_2\bigl(
|\omega_{ii'}|\,|\omega_{jj'}|
+
|\omega_{ij'}|\,|\omega_{i'j}|
\bigr),
\]
and hence, for some constant \(C_3>0\), we have
\begin{equation}
\sum_{\substack{a,b\in\mathcal I\\a\ne b}}
|\rho_{a,b}|
\le
C_3\left(
k+\sum_{1\le i<j\le k}|\omega_{ij}|
\right)^2.
\label{eq:rho-sum}
\end{equation}
Since \(I_{\{|\rho_{a,b}|>\eta\}}\le\eta^{-1}|\rho_{a,b}|\) and
\(|\operatorname{Cov}(I_A,I_B)|\le1/4\), it follows that, for some
constant \(C_{4,\eta}>0\),
\begin{equation}
\sup_{t\in\mathbb R}
\operatorname{Var}\{\widehat{\widetilde{H}}_N(t)\}
\le
\frac{1}{4N}
+
\frac{C_{1,\eta}}{\sqrt n}
+
\frac{C_{4,\eta}}{N^2}
\left(
k+\sum_{1\le i<j\le k}|\omega_{ij}|
\right)^2.
\label{eq:variance-Htilde}
\end{equation}
Suppose first that
\(\sum_{1\le i<j\le k}|\omega_{ij}|=o(k^2)\). Then, for every
\(t\in\mathbb R\),
\begin{equation}
\widehat{\widetilde{H}}_N(t)-\overline{\widetilde{H}}_{N}(t)
\to 0 \quad \text{ in probability}.
\label{eq:pointwise-Htilde}
\end{equation}
Fix \(t\in\mathbb R\), \(M>4(|t|+1)\), and \(L>0\). Choose
\(a_{n,k}\downarrow0\) such that
\(\Pr(R_{n,k}>a_{n,k})\longrightarrow0\), and put
\[
\mathcal E_{n,k}(L)
=
\left\{
R_{n,k}\le a_{n,k},\ 
D_{n,k}\le L\sqrt{\frac{\log k}{n}}
\right\}.
\]
On \(\mathcal E_{n,k}(L)\),
\[
|T_{ij}-\widetilde T_{ij}|
\le
a_{n,k}
+
LM\sqrt{\frac{\log k}{n}}
=:h_{n,k}(M,L)
\]
whenever \(|\mu_{ij}|\le M\). Hence,
\[
I_{\{T_{ij}\le t\}}\ne I_{\{\widetilde T_{ij}\le t\}}
\quad\Longrightarrow\quad
|\widetilde T_{ij}-t|\le h_{n,k}(M,L).
\]
Let
\(
\sigma_{ij}^2=1+\theta_{ij}^2.
\)
Since
\[
Z_{ij}
=
\frac{1}{\sqrt n}
\sum_{\ell=1}^n
\frac{
\varepsilon_{\ell i}\varepsilon_{\ell j}
-E(\varepsilon_{\ell i}\varepsilon_{\ell j})
}{
\sqrt{\delta_{ii}\delta_{jj}}
},
\]
where the standardized summands have uniformly bounded third absolute
moments, the univariate Berry--Esseen theorem gives, for some constant
$C_{\mathrm{BE}}>0$,
\[
\sup_{(i,j)\in\mathcal I}\sup_{x\in\mathbb R}
\left|
\Pr(\widetilde T_{ij}\le x)
-
\Phi\left(
\frac{x+\mu_{ij}}{\sigma_{ij}}
\right)
\right|
\le
\frac{C_{\mathrm{BE}}}{\sqrt n}.
\]
Since $1\le\sigma_{ij}^2\le2$, it follows that, for some constants
$C_5,C_6>0$,
\begin{equation}
\sup_{(i,j)\in\mathcal I}\sup_{u\in\mathbb R}
\Pr\bigl(|\widetilde T_{ij}-u|\le h\bigr)
\le
C_5h+\frac{C_6}{\sqrt n},
\qquad h>0.
\label{eq:anti-concentration}
\end{equation}
For all sufficiently large \(n\) and \(k\), on
\(\mathcal E_{n,k}(L)\),
\[
I_{\{T_{ij}\le t\}}\ne I_{\{\widetilde T_{ij}\le t\}},
\quad
|\mu_{ij}|>M
\quad\Longrightarrow\quad
|Z_{ij}|\ge\frac{M}{2}.
\]
Indeed, \(a_{n,k}\le1\) and
\(L\sqrt{\log k/n}\le1/4\) for all sufficiently large \(n\) and \(k\).
Since \(\operatorname{Var}(Z_{ij})=1+\theta_{ij}^{2}\le2\),
Chebyshev's inequality yields
\[
\sup_{\{(i,j):\,|\mu_{ij}|>M\}}
\Pr\left(|Z_{ij}|\ge\frac{M}{2}\right)
\le
\frac{8}{M^2}.
\]
Let
\[
\Delta_N(t)
=
\frac{1}{N}
\sum_{(i,j)\in\mathcal I}
\left|
I_{\{T_{ij}\le t\}}-I_{\{\widetilde T_{ij}\le t\}}
\right|.
\]
Then
\[
E\!\left[
\Delta_N(t)I_{\mathcal E_{n,k}(L)}
\right]
\le
C_5h_{n,k}(M,L)
+\frac{C_6}{\sqrt n}
+\frac{8}{M^2}.
\]
Therefore, for every \(\varepsilon>0\),
\[
\Pr\{\Delta_N(t)>\varepsilon\}
\le
\Pr\{\mathcal E_{n,k}(L)^c\}
+
\frac{1}{\varepsilon}
\left\{
C_5h_{n,k}(M,L)
+\frac{C_6}{\sqrt n}
+\frac{8}{M^2}
\right\}.
\]
Letting first $\min\{n,k\}\to\infty$, then \(L\to\infty\), and finally
\(M\to\infty\), gives \(\Delta_N(t)\to 0 \text{ in probability}\). As
\(0\le\Delta_N(t)\le1\),
\[
|\bar{H}_N(t)-\overline{\widetilde{H}}_{N}(t)|
\le
E\{\Delta_N(t)\}\longrightarrow0.
\]
Together with \eqref{eq:pointwise-Htilde},
\(\widehat{H}_N(t)-\bar{H}_N(t)\to 0 \text{ in probability}\), \(t\in\mathbb R\). The same
argument, with \(\le\) replaced by \(<\), gives
\(\widehat{H}_N(t-)-\bar{H}_N(t-)\to 0 \text{ in probability}\). For \(x\in(0,1)\), let \(z_x=\Phi^{-1}(1-x/2)\). Then
\[
\widehat{F}_N(x)=1+\widehat{H}_N(-z_x)-\widehat{H}_N(z_x-),
\]
and the same identity holds for \(\bar{F}_N (x)\). Thus,
\(\widehat{F}_N(x)-\bar{F}_N (x)\to 0 \text{ in probability}\), \(x\in(0,1)\). The cases
\(x=0\) and \(x=1\) are immediate. \\[0.05 in]
Now assume condition \eqref{C2}, and define
\[
\gamma_{n,k}
=
\frac{1}{N}
\sum_{(i,j)\in\mathcal I}
\min\left\{
1,\sqrt{\log k}\,|\theta_{ij}|
\right\}.
\]
For all sufficiently large \(k\),
\[
\sum_{1\le i<j\le k}|\omega_{ij}|
\le
c_0\sum_{(i,j)\in\mathcal I}|\theta_{ij}|
\le
c_0N\gamma_{n,k}
=
o(N).
\]
Hence, the right-hand side of \eqref{eq:variance-Htilde} converges to zero uniformly in
\(t\). Each \(\widetilde T_{ij}\) has a continuous distribution, so
\(\overline{\widetilde{H}}_{N}\) is continuous. A finite-grid argument and
Chebyshev's inequality now give
\begin{equation}
\|\widehat{\widetilde{H}}_N-\overline{\widetilde{H}}_{N}\|_\infty
\to 0 \quad \text{ in probability}.
\label{eq:uniform-Htilde}
\end{equation}
Fix \(0<h<2L\), and let
\[
\mathcal F_{n,k}(h,L)
=
\left\{
R_{n,k}\le\frac h2,\ 
D_{n,k}\le L\sqrt{\frac{\log k}{n}}
\right\}.
\]
On \(\mathcal F_{n,k}(h,L)\),
\[
\sqrt{\log k}\,|\theta_{ij}|
\le
\frac{h}{2L}
\quad\Longrightarrow\quad
|T_{ij}-\widetilde T_{ij}|\le h.
\]
Moreover,
\[
\frac{1}{N}
\sum_{(i,j)\in\mathcal I}
I_{\left\{
\sqrt{\log k}\,|\theta_{ij}|>h/(2L)
\right\}}
\le
\frac{2L}{h}\gamma_{n,k}.
\]
Indeed, for every pair satisfying $|T_{ij}-\widetilde T_{ij}|\le h$,
\[
\left|
I_{\{T_{ij}\le t\}}-I_{\{\widetilde T_{ij}\le t\}}
\right|
\le
I_{\{t-h<\widetilde T_{ij}\le t+h\}}.
\]
For the remaining pairs, the absolute difference of the indicators is
bounded by one. Therefore, on $\mathcal F_{n,k}(h,L)$,
\[
\|\widehat{H}_N-\widehat{\widetilde{H}}_N\|_\infty
\le
\sup_{t\in\mathbb R}
\{\widehat{\widetilde{H}}_N(t+h)-\widehat{\widetilde{H}}_N(t-h)\}
+
\frac{2L}{h}\gamma_{n,k}.
\]
By adding and subtracting $\overline{\widetilde{H}}_{N}$ and using
\eqref{eq:anti-concentration},
\[
\sup_{t\in\mathbb R}
\{\widehat{\widetilde{H}}_N(t+h)-\widehat{\widetilde{H}}_N(t-h)\}
\le
2\|\widehat{\widetilde{H}}_N-\overline{\widetilde{H}}_{N}\|_\infty
+
C_5h+\frac{C_6}{\sqrt n}.
\]
Thus, for every \(\varepsilon>0\),
\[
\Pr\left(
\|\widehat{H}_N-\widehat{\widetilde{H}}_N\|_\infty>\varepsilon
\right)
\le
\Pr\{\mathcal F_{n,k}(h,L)^c\} 
+
\Pr\left(
2\|\widehat{\widetilde{H}}_N-\overline{\widetilde{H}}_{N}\|_\infty
+
C_5h+\frac{C_6}{\sqrt n}
+
\frac{2L}{h}\gamma_{n,k}
>
\varepsilon
\right).
\]
Using \eqref{eq:uniform-Htilde} and condition \eqref{C2}, and then letting first
$\min\{n,k\}\to\infty$, next \(L\to\infty\), and finally \(h\downarrow0\),
yields \(\|\widehat{H}_N-\widehat{\widetilde{H}}_N\|_\infty\to 0 \: \: \text{ in probability}\). Since this
norm is bounded by one,
\[
\|\bar{H}_N-\overline{\widetilde{H}}_{N}\|_\infty
\le
E\|\widehat{H}_N-\widehat{\widetilde{H}}_N\|_\infty
\longrightarrow0.
\]
Hence,
\[
\|\widehat{H}_N-\bar{H}_N\|_\infty\to 0 \quad \text{ in probability}.
\]
Corollary~\ref{l_pval} now gives
\[
\|\widehat{F}_N-\bar{F}_N\|_\infty\to 0 \quad \text{ in probability}.
\]
\hfill\(\square\)
\\[0.05 in]
We will now present a few lemmas that facilitate the proof of
Corollary \ref{cor1}.
\begin{lemma}\label{Fpij}
Under Condition~\eqref{C1} and if  \(\log k=o(\sqrt n)\), we have
\[
\sup_{\lambda\in[0,1]}
\frac{1}{N}
\sum_{(i,j)\in I_0(k)}
\left|
\Pr(p_{ij}\le\lambda)-\lambda
\right|
=o(1),
\]
where \(N=k(k-1)/2\) is the total number of hypotheses.
\end{lemma}
\noindent\textbf{Proof :} For \((i,j)\in I_0(k)\), define
\(
Z_{ij}
=
U_{ij}/\sqrt{\delta_{ii}\delta_{jj}}.
\)
Since \(\mu_{ij}=0\) under \(H_{0,ij}\), relation \eqref{eq:T-tildeT} in the
proof of Theorem \ref{th1} gives
\[
\Delta_{n,k}
:=
\max_{(i,j)\in I_0(k)}
|T_{ij}-Z_{ij}|
\le R_{n,k}
\to 0 \quad \text{ in probability}.
\]
Choose a deterministic sequence \(a_{n,k}\downarrow0\) such that
\(
\Pr(\Delta_{n,k}>a_{n,k})\longrightarrow0.
\)
Under \(H_{0,ij}\), the variables
\(
\varepsilon_{li}/\sqrt{\delta_{ii}}
\) and \(
\varepsilon_{lj}/\sqrt{\delta_{jj}}
\)
are independent standard Gaussian random variables. Hence, the
Berry--Esseen theorem gives
\[
\sup_{(i,j)\in I_0 (k)}
\sup_{x\in\mathbb R}
\left|
\Pr(Z_{ij}\le x)-\Phi(x)
\right|
\le
\frac{C}{\sqrt n}.
\]
Therefore,
\begin{equation}
\sup_{(i,j)\in I_0 (k) }
\sup_{x\in\mathbb R}
\left|
\Pr(T_{ij}\le x)-\Phi(x)
\right|
\le
\Pr(\Delta_{n,k}>a_{n,k})
+
\frac{a_{n,k}}{\sqrt{2\pi}}
+
\frac{C}{\sqrt n} 
=o(1).
\label{eq:null-normal-approx}
\end{equation} The same bound holds for the left limits of the distribution functions. For \(0<\lambda<1\), let
\(
z_\lambda=\Phi^{-1}(1-\lambda/2).
\)
Then
\[
\Pr(p_{ij}\le\lambda)
=
\Pr(T_{ij}\le-z_\lambda)
+
\Pr(T_{ij}\ge z_\lambda).
\]
It follows from \eqref{eq:null-normal-approx} that
\[
\sup_{(i,j)\in I_0 (k) }
\sup_{\lambda\in[0,1]}
\left|
\Pr(p_{ij}\le\lambda)-\lambda
\right|
=o(1).
\]
Averaging over \(I_0 (k) \) completes the proof.
\hfill\(\square\)

\begin{lemma}\label{altd}

If \(X\sim N(a,\sigma^2)\), where \(a\in\mathbb R\) and
\(\sigma^2\ge1\), then the two-sided p-value
\[
P=2\bigl(1-\Phi(|X|)\bigr)
\]
has a concave distribution function on \([0,1]\).

\end{lemma}

\noindent\textbf{Proof :} For \(0<\lambda<1\), let
\(
z_\lambda=\Phi^{-1}(1-\lambda/2).
\)
Then
\[
F_P(\lambda)
=
1-\Phi\left(\frac{z_\lambda-a}{\sigma}\right)
+
\Phi\left(\frac{-z_\lambda-a}{\sigma}\right).
\]
Differentiating yields
\[
F_P'(\lambda)
=
\frac{
\phi\left((z_\lambda-a)/\sigma\right)
+
\phi\left((z_\lambda+a)/\sigma\right)
}{
2\sigma\phi(z_\lambda)
}.
\]
Equivalently,
\[
F_P'(\lambda)
=
\frac{1}{\sigma}
\exp\left\{
-\frac{a^2}{2\sigma^2}
+
\frac12\left(1-\frac{1}{\sigma^2}\right)z_\lambda^2
\right\}
\cosh\left(\frac{az_\lambda}{\sigma^2}\right).
\]
Since \(\sigma^2\ge1\), this expression is nondecreasing in
\(z_\lambda\ge0\). Since \(z_\lambda\) is decreasing in \(\lambda\),
\(F_P'(\lambda)\) is nonincreasing in \(\lambda\). Hence,
\(F_P\) is concave.
\hfill\(\square\)
\begin{corollary}\label{l_altcdf}
Suppose that the assumptions of Theorem \ref{th1}, including Condition~\eqref{C2}, hold. Define
\[
F_{1,n,k}(\lambda)
=
\frac{1}{N}
\sum_{(i,j)\in I_1 (k)}
\Pr(p_{ij}\le\lambda),
\qquad
0\le\lambda\le1.
\]
Then there exists a sequence of concave functions
\(F_{1,n,k}^{\circ}\) on \([0,1]\) such that
\[
\|F_{1,n,k}-F_{1,n,k}^{\circ}\|_\infty
\longrightarrow0.
\]
\end{corollary}

\noindent\textbf{Proof :}
For \((i,j)\in I_1 (k)\), let
\[
\theta_{ij}
=
\frac{\omega_{ij}}{\sqrt{\omega_{ii}\omega_{jj}}},
\qquad
\mu_{ij}=\sqrt n\,\theta_{ij},
\]
and let \(V_{ij}\) have distribution
\[
N\bigl(-\mu_{ij},1+\theta_{ij}^{2}\bigr).
\]
Define
\[
G_{ij}^{\circ}(\lambda)
=
\Pr\left\{
2\bigl(1-\Phi(|V_{ij}|)\bigr)\le\lambda
\right\},
\]
and
\[
F_{1,n,k}^{\circ}(\lambda)
=
\frac{1}{N}
\sum_{(i,j)\in I_1 (k)}
G_{ij}^{\circ}(\lambda).
\]
By Lemma \ref{altd}, each \(G_{ij}^{\circ}\) is concave. Hence,
\(F_{1,n,k}^{\circ}\) is concave. \\[0.05 in] 
Let \(H_{ij}\), \(\widetilde H_{ij}\), and \(H_{ij}^{\circ}\) denote
the distribution functions of \(T_{ij}\), \(\widetilde T_{ij}\), and
\(V_{ij}\), respectively. The Berry--Esseen theorem yields
\begin{equation}
\max_{(i,j)\in\mathcal I}
\|\widetilde H_{ij}-H_{ij}^{\circ}\|_\infty
\le
\frac{C}{\sqrt n}.
\label{eq:BE-alternative}
\end{equation}
Fix \(0<h<2L\), and let
\[
\mathcal F_{n,k}(h,L)
=
\left\{
R_{n,k}\le\frac h2,\ 
D_{n,k}\le L\sqrt{\frac{\log k}{n}}
\right\}.
\]
For every \((i,j)\) satisfying
\(
\sqrt{\log k}\,|\theta_{ij}|
\le
\frac{h}{2L},
\)
relation \eqref{eq:T-tildeT} gives
\(
|T_{ij}-\widetilde T_{ij}|
\le h
\)
on \(\mathcal F_{n,k}(h,L)\). Therefore,
\[
\frac{1}{N}
\sum_{(i,j)\in I_1 (k)}
\|H_{ij}-\widetilde H_{ij}\|_\infty
\le
\Pr\{\mathcal F_{n,k}(h,L)^c\}
+
Ch+\frac{C}{\sqrt n} 
+
\frac{1}{N}
\sum_{(i,j)\in\mathcal I}
I_{\left\{
\sqrt{\log k}\,|\theta_{ij}|>h/(2L)
\right\}}.
\]
By condition \eqref{C2},
\[
\frac{1}{N}
\sum_{(i,j)\in\mathcal I}
I_{\left\{
\sqrt{\log k}\,|\theta_{ij}|>h/(2L)
\right\}}
\le
\frac{2L}{h}\gamma_{n,k},
\]
where
\[
\gamma_{n,k}
=
\frac{1}{N}
\sum_{(i,j)\in\mathcal I}
\min\left\{
1,\sqrt{\log k}\,|\theta_{ij}|
\right\}.
\]
Letting first $\min\{n,k\}\to\infty$, then \(L\to\infty\), and finally
\(h\downarrow0\), we obtain
\begin{equation}
\frac{1}{N}
\sum_{(i,j)\in I_1 (k)}
\|H_{ij}-\widetilde H_{ij}\|_\infty
\longrightarrow0.
\label{eq:T-tilde-average}
\end{equation}
For arbitrary real-valued random variables \(X\) and \(Y\),
\[
\sup_{\lambda\in[0,1]}
\left|
\Pr\{2(1-\Phi(|X|))\le\lambda\}
-
\Pr\{2(1-\Phi(|Y|))\le\lambda\}
\right|
\le
2\|F_X-F_Y\|_\infty.
\]
Consequently, by \eqref{eq:BE-alternative} and \eqref{eq:T-tilde-average},
\[
\begin{aligned}
\|F_{1,n,k}-F_{1,n,k}^{\circ}\|_\infty
&\le
\frac{2}{N}
\sum_{(i,j)\in I_1 (k)}
\|H_{ij}-\widetilde H_{ij}\|_\infty \\
&\quad+
\frac{2}{N}
\sum_{(i,j)\in I_1 (k)}
\|\widetilde H_{ij}-H_{ij}^{\circ}\|_\infty \\
&\longrightarrow0.
\end{aligned}
\]
This proves the result.
\hfill\(\square\)
 \\[0.05 in]
\noindent\textbf{Proof of Corollary \ref{cor1} : }
Let
$$
A_{0,n,k}(\lambda)
=
\frac{1}{N}\sum_{(i,j)\in I_0(k)}\Pr(p_{ij}>\lambda).
$$
By Lemma \ref{Fpij},
\begin{equation}\label{null_approx}
\sup_{\lambda\in[0,1]}
\left|A_{0,n,k}(\lambda)-\pi_0(1-\lambda)\right|
=o(1).
\end{equation}
Moreover,
$$
1-F_{n,k}(\lambda)
=
A_{0,n,k}(\lambda)
+
\pi_1\overline F_{1,n,k}(\lambda).
$$
Hence, for every fixed $\lambda\in[0,1)$,
$$
\begin{aligned}
\hat\pi_0(\lambda)
-
\left\{
\pi_0+
\pi_1\frac{\overline F_{1,n,k}(\lambda)}{1-\lambda}
\right\}
&=
\frac{F_{n,k}(\lambda)-\widehat{F}_N(\lambda)}{1-\lambda}\\
&\quad+
\frac{A_{0,n,k}(\lambda)-\pi_0(1-\lambda)}{1-\lambda}.
\end{aligned}
$$
The first term converges to zero in probability by Theorem
\ref{th1}, and the second converges to zero by
\eqref{null_approx}. Therefore,
$$
\hat\pi_0(\lambda)
-
\left\{
\pi_0+
\pi_1\frac{\overline F_{1,n,k}(\lambda)}{1-\lambda}
\right\}
\longrightarrow0
$$
in probability. Since $\overline F_{1,n,k}(\lambda)\ge0$, the
Schweder--Spj{\o}tvoll estimator is asymptotically biased upwards. For the final assertion, define
$$
F_{0,n,k}(\lambda)=\frac{1}{N}\sum_{(i,j)\in I_0(k)}
\Pr(p_{ij}\le\lambda),
\qquad
F_{1,n,k}(\lambda)=\frac{1}{N}\sum_{(i,j)\in I_1(k)}
\Pr(p_{ij}\le\lambda).
$$
Then $F_{n,k}=F_{0,n,k}+F_{1,n,k}$. By Lemma \ref{Fpij},
$$
\sup_{\lambda\in[0,1]}
\left|F_{0,n,k}(\lambda)-\pi_0\lambda\right|
\longrightarrow0.
$$
By Corollary \ref{l_altcdf}, there exists a sequence of concave functions
$F_{1,n,k}^{\circ}$ on $[0,1]$ such that
$$
\|F_{1,n,k}-F_{1,n,k}^{\circ}\|_\infty
\longrightarrow0.
$$
Define
$$
F_{n,k}^{\circ}(\lambda)
=
\pi_0\lambda+F_{1,n,k}^{\circ}(\lambda).
$$
Since $\lambda\mapsto\pi_0\lambda$ is linear and
$F_{1,n,k}^{\circ}$ is concave, $F_{n,k}^{\circ}$ is concave.
Furthermore,
$$
\begin{aligned}
\|F_{n,k}-F_{n,k}^{\circ}\|_\infty
&\le
\sup_{\lambda\in[0,1]}
\left|F_{0,n,k}(\lambda)-\pi_0\lambda\right|
+
\|F_{1,n,k}-F_{1,n,k}^{\circ}\|_\infty\\
&\longrightarrow0.
\end{aligned}
$$
Consequently, every pointwise limit of $F_{n,k}$ is concave.
\hfill$\square$
\bibliography{references}

@article{storey2003statistical,
  title={Statistical significance for genomewide studies},
  author={Storey, John D and Tibshirani, Robert},
  journal={Proceedings of the National Academy of Sciences},
  volume={100},
  number={16},
  pages={9440--9445},
  year={2003},
  publisher={National Academy of Sciences},
   url ={https://doi.org/10.1073/pnas.1530509100}
}

@article{dey2024limiting,
  title={On limiting behaviors of stepwise multiple testing procedures},
  author={Dey, Monitirtha},
  journal={Statistical Papers},
  volume={65},
  number={9},
  pages={5691--5717},
  year={2024},
  publisher={Berlin, Heidelberg: Springer},
  url={https://doi.org/10.1007/s00362-024-01613-6}
}

@article{FDR2007,
author = {Helmut Finner and Thorsten Dickhaus and Markus Roters},
title = {{Dependency and false discovery rate: Asymptotics}},
volume = {35},
journal = {The Annals of Statistics},
number = {4},
publisher = {Institute of Mathematical Statistics},
pages = {1432 -- 1455},
year = {2007},

URL = {https://doi.org/10.1214/009053607000000046}
}

@book{dickhaus2014simultaneous,
  title={Simultaneous statistical inference},
  author={Dickhaus, Thorsten},
  journal={With applications in the life sciences. Springer},
  year={2014},
url={https://link.springer.com/book/10.1007/978-3-642-45182-9},
  publisher={Springer}
}

@article{hengartner1999,
 ISSN = {00905364, 21688966},
 URL = {https://doi.org/10.1214/aos/1176324534},
 author = {Nicolas W. Hengartner and Philip B. Stark},
 journal = {The Annals of Statistics},
 number = {2},
 pages = {525--550},
 publisher = {Institute of Mathematical Statistics},
 title = {Finite-Sample Confidence Envelopes for Shape-Restricted Densities},
 volume = {23},
 year = {1995}
}

@article{swanepoel1999,
author = {Jan W. H. Swanepoel},
title = {{The limiting behavior of a modified maximal symmetric $2s$-spacing with applications}},
volume = {27},
journal = {The Annals of Statistics},
number = {1},
publisher = {Institute of Mathematical Statistics},
pages = {24 -- 35},
year = {1999},
URL = {https://doi.org/10.1214/aos/1018031099}
}

@article{genovese2004stochastic,
author = {Christopher Genovese and Larry Wasserman},
title = {{A stochastic process approach to false discovery control}},
volume = {32},
journal = {The Annals of Statistics},
number = {3},
publisher = {Institute of Mathematical Statistics},
pages = {1035 -- 1061},
year = {2004},
URL = {https://doi.org/10.1214/009053604000000283}
}

@article{sun2012scaled,
  title={Scaled sparse linear regression},
  author={Sun,Tingni and Zhang,Cun-Hui},
  journal={Biometrika},
  volume={99},
  number={4},
  pages={879--898},
  year={2012},
  publisher={Oxford University Press},
  issn = {0006-3444},
    url = {https://doi.org/10.1093/biomet/ass043}
}

@article{cancer2012comprehensive,
  author  = {Koboldt, Daniel C. and
             Fulton, Robert S. and
             McLellan, Michael D. and others},
  title   = {Comprehensive Molecular Portraits of Human Breast Tumours},
  journal = {Nature},
  year    = {2012},
  volume  = {490},
  number  = {7418},
  pages   = {61--70},
  url     = {https://doi.org/10.1038/nature11412}
}

@article{raskutti2008model,
  title={Model Selection in Gaussian Graphical Models: High-Dimensional Consistency of $ l_1$ regularized MLE},
  author={Raskutti, Garvesh and Yu, Bin and Wainwright, Martin J and Ravikumar, Pradeep},
  journal={Advances in Neural Information Processing Systems},
  volume={21},
  year={2008}
}

@article{Lancaster1957,
  title={Some properties of the bivariate normal distribution considered in the form of a contingency table},
  author={Lancaster, Henry Oliver},
  journal={Biometrika},
  volume={44},
  number={1/2},
  pages={289--292},
  year={1957},
  url={https://doi.org/10.1093/biomet/44.1-2.289}
}

@article{lingjaerde2021tailored,
  title={Tailored graphical lasso for data integration in gene network reconstruction},
  author={Lingj{\ae}rde, Camilla and Lien, Tonje G and Borgan, {\O}rnulf and Bergholtz, Helga and Glad, Ingrid K},
  journal={BMC bioinformatics},
  volume={22},
  number={1},
  pages={498},
  year={2021},
  url={https://doi.org/10.1186/s12859-021-04413-z},
  publisher={Springer}
}

@article{10.3150/25-BEJ1858,
author = {Thorsten Dickhaus and Ruth Heller and Anh-Tuan Hoang and Yosef Rinott},
title = {{A procedure for multiple testing of partial conjunction hypotheses based on a hazard rate inequality}},
volume = {32},
journal = {Bernoulli},
number = {1},
publisher = {Bernoulli Society for Mathematical Statistics and Probability},
pages = {274 -- 298},
year = {2026},
URL = {https://doi.org/10.3150/25-BEJ1858}
}

@article{bauer2012pair,
  title={Pair-copula constructions for non-Gaussian DAG models},
  author={Bauer, Alexander and Czado, Claudia and Klein, Thomas},
  journal={Canadian Journal of Statistics},
  volume={40},
  number={1},
  pages={86--109},
  year={2012},
  url = {https://doi.org/10.1002/cjs.10131},
  publisher={Wiley Online Library}
}

@article{li2025novel,
  title={A novel approach for estimating multi-attribute Gaussian copula graphical models},
  author={Li, Lijie and Yu, Yang and Liang, Wanfeng and Zou, Feng},
  journal={Statistics \& Probability Letters},
  volume={222},
  pages={110413},
  year={2025},
  publisher={Elsevier},
issn = {0167-7152},
url ={https://doi.org/10.1016/j.spl.2025.110413}
}

@article{hermes2024copula,
author = {Sjoerd Hermes and Joost van Heerwaarden and Pariya Behrouzi},
title = {Copula Graphical Models for Heterogeneous Mixed Data},
journal = {Journal of Computational and Graphical Statistics},
volume = {33},
number = {3},
pages = {991--1005},
year = {2024},
publisher = {Taylor \& Francis},
URL ={https://doi.org/10.1080/10618600.2023.2289545
}
}

@article{dobra2011copula,
  title={Copula Gaussian graphical models and their application to modeling functional disability data},
  author={Dobra, Adrian and Lenkoski, Alex},
  year={2011},
 journal = {The Annals of Applied Statistics},
 number = {2A},
 pages = {969--993},
 publisher = {Institute of Mathematical Statistics},
  volume = {5}, 
URL = {https://doi.org/10.1214/10-AOAS397}
}

@article{behrouzi2019detecting,
  title={Detecting epistatic selection with partially observed genotype data by using copula graphical models},
  author={Behrouzi, Pariya and Wit, Ernst C},
  journal={Journal of the Royal Statistical Society Series C: Applied Statistics},
  volume={68},
  number={1},
  pages={141--160},
  year={2019},
 issn = {0035-9254},
    url = {https://doi.org/10.1111/rssc.12287},
  publisher={Oxford University Press}
}

@INPROCEEDINGS{6288344,
  author={Yu, Hang and Dauwels, Justin and Wang, Xueou},
  booktitle={2012 IEEE International Conference on Acoustics, Speech and Signal Processing (ICASSP)}, 
  title={Copula Gaussian graphical models with hidden variables}, 
  year={2012},
  pages={2177-2180},
  organization={IEEE},
  doi={10.1109/ICASSP.2012.6288344}}

@article{hall2009using,
 title={Using the bootstrap to quantify the authority of an empirical ranking},
  author={Peter Hall and Hugh Miller},
  journal={Annals of Statistics},
  year={2009},
  volume={37},
  publisher = {Institute of Mathematical Statistics},
  pages={3929-3959},
  url={https://doi.org/10.1214/09-AOS699}
}

@article{neumann2021estimating,
  title = {Estimating the proportion of true null hypotheses under dependency: A marginal bootstrap approach},
journal = {Journal of Statistical Planning and Inference},
volume = {210},
pages = {76-86},
year = {2021},
issn = {0378-3758},
url = {https://doi.org/10.1016/j.jspi.2020.04.011},
author = {André Neumann and Taras Bodnar and Thorsten Dickhaus},
}

@article{demko1984decay,
  title={Decay rates for inverses of band matrices},
  author={Demko, Stephen and Moss, William F and Smith, Philip W},
  journal={Mathematics of computation},
  volume={43},
  number={168},
  pages={491--499},
ISSN = {00255718, 10886842},
 URL = {https://doi.org/10.1090/S0025-5718-1984-0758197-9},
  year={1984}
}

@book {MR112166,
    AUTHOR = {Parzen, Emanuel},
     TITLE = {Modern probability theory and its applications},
    SERIES = {A Wiley Publication in Mathematical Statistics},
 PUBLISHER = {John Wiley \& Sons, Inc., New York-London},
      YEAR = {1960},
     PAGES = {xv+464},
   MRCLASS = {60.00},
  MRNUMBER = {112166},
MRREVIEWER = {F.\ L.\ Spitzer},
}

@book{billingsley2017probability,
  title={Probability and measure},
  author={Billingsley, Patrick},
  year={1995},
  publisher={John Wiley \& Sons}
}

@article{isserlis1918formula,
  title={On a formula for the product-moment coefficient of any order of a normal frequency distribution in any number of variables},
  author={Isserlis, Leon},
  journal={Biometrika},
  volume={12},
  number={1/2},
  pages={134--139},
  year={1918},
  publisher={JSTOR}
}

@article{liu2013gaussian,
 author = {Weidong Liu},
 journal = {The Annals of Statistics},
 number = {6},
 pages = {2948--2978},
 publisher = {Institute of Mathematical Statistics},
 title = {Gaussian graphical model estimation with false discovery rate control},
 volume = {41},
 year = {2013},
 ISSN = {00905364, 21688966},
 URL = {https://doi.org/10.1214/13-AOS1169}
}

@article{chen2019uniformly,
title = {Uniformly consistently estimating the proportion of false null hypotheses via Lebesgue–Stieltjes integral equations},
journal = {Journal of Multivariate Analysis},
volume = {173},
pages = {724-744},
year = {2019},
issn = {0047-259X},
url = {https://doi.org/10.1016/j.jmva.2019.06.003},
author = {Xiongzhi Chen}
}

@article{chen2025uniformly,
  title = {Uniformly consistent proportion estimation for composite hypotheses via integral equations: ``the case of {Gamma} random variables''},
  author={Chen, Xiongzhi},
  journal={Annals of the Institute of Statistical Mathematics},
  volume={77},
  number={4},
  pages={649--684}, 
  year={2025},
  publisher={Springer},
  url={https://doi.org/10.1007/s10463-025-00930-3}
}

@article{friedman2008sparse,
  title={Sparse inverse covariance estimation with the graphical lasso},
  author={Friedman, Jerome and Hastie, Trevor and Tibshirani, Robert},
  journal={Biostatistics},
  volume={9},
  number={3},
  pages={432--441},
  year={2008},
    url = {https://doi.org/10.1093/biostatistics/kxm045},
  publisher={Oxford University Press}
}

@article{d2008first,
  title={First-order methods for sparse covariance selection},
  author={d'Aspremont, Alexandre and Banerjee, Onureena and El Ghaoui, Laurent},
  journal={SIAM Journal on Matrix Analysis and Applications},
  volume={30},
  number={1},
  pages={56--66},
  year={2008},
URL = {https://doi.org/10.1137/060670985},
  publisher={SIAM}
}

@article{yuan2010high,
  title={High dimensional inverse covariance matrix estimation via linear programming},
  author={Yuan, Ming},
  journal={The Journal of Machine Learning Research},
  volume={11},
  pages={2261--2286},
  year={2010},
  publisher={JMLR. org},
  url= {http://jmlr.org/papers/v11/yuan10b.html}
}

@misc{zhang2010estimation,
  author = {Zhang, C.},
  title = {Estimation of Large Inverse Matrices and Graphical Model Selection},
  year = {2010},
  note = {Technical Report, Department of Statistics and Biostatistics, Rutgers University}
}

@article{meinshausen2006high,
  author = {Nicolai Meinshausen and Peter B{\"u}hlmann},
title = {{High-dimensional graphs and variable selection with the Lasso}},
volume = {34},
journal = {The Annals of Statistics},
number = {3},
publisher = {Institute of Mathematical Statistics},
pages = {1436 -- 1462},
year = {2006},
URL = {https://doi.org/10.1214/009053606000000281}
}

@article{liu2012high,
  title={High-dimensional semiparametric Gaussian copula graphical models},
  author={Liu, Han and Han, Fang and Yuan, Ming and Lafferty, John and Wasserman, Larry},
  year={2012},
 journal = {The Annals of Statistics},
 number = {4},
 pages = {2293--2326}, 
 publisher = {Institute of Mathematical Statistics},
 volume = {40},
URL = {https://doi.org/10.1214/12-AOS1037}
}

@article{cai2011constrained,
  title={A constrained l-1 minimization approach to sparse precision matrix estimation},
  author={Cai, Tony and Liu, Weidong and Luo, Xi},
  journal={Journal of the American Statistical Association},
  volume={106},
  number={494},
  pages={594--607},
  year={2011},
URL = { https://doi.org/10.1198/jasa.2011.tm10155
},
  publisher={Taylor \& Francis}
}

@article{xue2012regularized,
  ISSN = {00905364, 21688966},
 URL = {https://doi.org/10.1214/12-AOS1041},
 author = {Lingzhou Xue and Hui Zou},
 journal = {The Annals of Statistics},
 number = {5},
 pages = {2541--2571},
 publisher = {Institute of Mathematical Statistics},
 title = {Regularized rank-based estimation of high-dimensional nonparanormal graphical models},
 volume = {40},
url={https://doi.org/10.1214/12-AOS1041},
 year = {2012}
}

@book{lauritzen1996graphical,
  title={Graphical models},
  author={Lauritzen, Steffen L},
  volume={17},
  year={1996},
  publisher={Clarendon Press},
}

@article{langaas2005estimating,
  title={Estimating the proportion of true null hypotheses, with application to DNA microarray data},
  author={Langaas, Mette and Lindqvist, Bo Henry and Ferkingstad, Egil},
  journal={Journal of the Royal Statistical Society Series B: Statistical Methodology},
  volume={67},
  number={4},
  pages={555--572},
  year={2005},
  publisher={Oxford University Press},
   issn = {1369-7412},
    url = {https://doi.org/10.1111/j.1467-9868.2005.00515.x}
}

@article{storey2002direct,
   author = {John D. Storey},
 journal = {Journal of the Royal Statistical Society. Series B (Statistical Methodology)},
 number = {3},
 pages = {479--498},
 publisher = {[Royal Statistical Society, Wiley]},
 title = {A Direct Approach to False Discovery Rates},
 volume = {64},
 year = {2002},
  ISSN = {13697412, 14679868},
 URL = {https://doi.org/10.1111/1467-9868.00346}
}

@article{storey2004strong,
  title={Strong control, conservative point estimation and simultaneous conservative consistency of false discovery rates: a unified approach},
  author={Storey, John D and Taylor, Jonathan E and Siegmund, David},
  journal={Journal of the Royal Statistical Society Series B: Statistical Methodology},
  volume={66},
  number={1},
  pages={187--205},
  year={2004},
  publisher={Oxford University Press},
issn = {1369-7412},
    url = {https://doi.org/10.1111/j.1467-9868.2004.00439.x}
}

@article{benjamini2000adaptive,
  title={On the adaptive control of the false discovery rate in multiple testing with independent statistics},
  author={Benjamini, Yoav and Hochberg, Yosef},
  journal={Journal of educational and Behavioral Statistics},
  volume={25},
  number={1},
  pages={60--83},
  year={2000},
  publisher={Sage Publications Sage CA: Los Angeles, CA},
  url={ 
https://doi.org/10.3102/10769986025001060}
}

@article{patra2016estimation,
    author = {Patra, Rohit Kumar and Sen, Bodhisattva},
    title = {Estimation of a Two-component Mixture Model with Applications to Multiple Testing},
    journal = {Journal of the Royal Statistical Society Series B: Statistical Methodology},
    volume = {78},
    number = {4},
    pages = {869-893},
    year = {2016},
    month = {01},
    issn = {1369-7412},
    url = {https://doi.org/10.1111/rssb.12148}
}

@article{jin2008proportion,
    author = {Jin, Jiashun},
    title = {Proportion of Non-Zero Normal Means: Universal Oracle Equivalences and Uniformly Consistent Estimators},
    journal = {Journal of the Royal Statistical Society Series B: Statistical Methodology},
    volume = {70},
    number = {3},
    pages = {461-493},
    year = {2008},
    month = {04},
    issn = {1369-7412},
    url = {https://doi.org/10.1111/j.1467-9868.2007.00645.x}
}

@article{jin2007estimating,
author = {Jiashun Jin and T. Tony Cai},
title = {Estimating the Null and the Proportion of Nonnull Effects in Large-Scale Multiple Comparisons},
journal = {Journal of the American Statistical Association},
volume = {102},
number = {478},
pages = {495--506},
year = {2007},
publisher = {Taylor \& Francis},
URL = { https://doi.org/10.1198/016214507000000167
}
}

@article{benjamini1995controlling,
  title={Controlling the false discovery rate: a practical and powerful approach to multiple testing},
  author={Benjamini, Yoav and Hochberg, Yosef},
  journal={Journal of the Royal statistical society: series B (Methodological)},
  volume={57},
  number={1},
  pages={289--300},
  year={1995},
  publisher={Wiley Online Library},
  url={https://doi.org/10.1111/j.2517-6161.1995.tb02031.x}
}

@article{schweder1982plots,
  ISSN = {00063444, 14643510},
 author = {T. Schweder and E. Spjøtvoll},
 journal = {Biometrika},
 number = {3},
 pages = {493--502},
 publisher = {[Oxford University Press, Biometrika Trust]},
 title = {Plots of P-Values to Evaluate Many Tests Simultaneously},
  url={https://doi.org/10.1093/biomet/69.3.493},
 volume = {69},
 year = {1982}
}

@article{benjamini2001control,
  title={The control of the false discovery rate in multiple testing under dependency},
  author={Benjamini, Yoav and Yekutieli, Daniel},
  journal={The annals of statistics},
  volume={29},
  number={4},
  URL = { https://doi.org/10.1214/aos/1013699998},
  pages={1165--1188},
  year={2001},
  publisher={Institute of Mathematical Statistics}
}

@article{sarkar2002some,
  title={Some results on false discovery rate in stepwise multiple testing procedures},
  author={Sarkar, Sanat K},
  journal={The Annals of Statistics},
  volume={30},
  number={1},
  pages={239--257},
  year={2002},
  url={ https://doi.org/10.1214/aos/1015362192},
  publisher={Institute of Mathematical Statistics}
}

@article{efron2001empirical,
  title={Empirical Bayes analysis of a microarray experiment},
  author={Efron, Bradley and Tibshirani, Robert and Storey, John D and Tusher, Virginia},
  journal={Journal of the American statistical association},
  volume={96},
  number={456},
  pages={1151--1160},
  year={2001},
  URL = {https://doi.org/10.1198/016214501753382129},
  publisher={Taylor \& Francis}
}

@article{das2021bound,
  title={Bound on FWER for correlated normal},
  author={Das, Nabaneet and Bhandari, Subir Kumar},
  journal={Statistics \& Probability Letters},
  volume={168},
  pages={108943},
  year={2021},
  url = {https://doi.org/10.1016/j.spl.2020.108943},
  publisher={Elsevier}
}

@article{li2021ggm,
    title = {GGM Knockoff Filter: False Discovery Rate Control for Gaussian Graphical Models}, 
    author = {Li, Jinzhou and Maathuis, Marloes H.},
    journal = {Journal of the Royal Statistical Society Series B: Statistical Methodology},
    volume = {83},
    number = {3},
    pages = {534-558},
    year = {2021},
    month = {07},
    issn = {1369-7412},
    url = {https://doi.org/10.1111/rssb.12430}
}

@article{zhou2022reproducible,
title = {Reproducible learning in large-scale graphical models},
journal = {Journal of Multivariate Analysis},
volume = {189},
pages = {104934},
year = {2022},
issn = {0047-259X},

url = {https://doi.org/10.1016/j.jmva.2021.104934},
author = {Jia Zhou and Yang Li and Zemin Zheng and Daoji Li}
}

@article{zhou2026reproducible,
   title={Reproducible Learning in Large-Scale Multiple Graphical Models},
  author={Zhou, Jia and Pan, Guangming and Zheng, Zeming and Tan, Changchun},
  journal={Statistica Sinica},
  volume={36},
  number={4},
  year={2026},
  doi={10.5705/ss.202023.0099}
}

\end{document}